%% file: main.tex
\documentclass{article}
\input{style}
\usepackage{graphicx} 

\begin{document}

\title{Quantum Metropolis Sampling via  Weak Measurement}

\author[1]{Jiaqing Jiang~\thanks{jiaqingjiang95@gmail.com}}
\affil[1]{Department of Computing and Mathematical Sciences, California Institute of Technology, Pasadena, CA, USA}

\author[2]{Sandy Irani~\thanks{irani@ics.uci.edu}}
\affil[2]{Computer Science Department, University of California, Irvine, CA, USA}
\date{}
\maketitle

\begin{abstract}

Gibbs sampling is a crucial computational technique  used in physics, statistics, and many other scientific fields. For classical Hamiltonians, the most commonly used Gibbs sampler  is the  Metropolis algorithm, known for having
the Gibbs state as its unique fixed point. 
For quantum Hamiltonians, designing provably correct Gibbs samplers has been more challenging.
~\cite{temme2011quantum} introduced a novel method that uses quantum phase estimation (QPE) and
the Marriot-Watrous rewinding technique to mimic the classical Metropolis algorithm for quantum Hamiltonians.
The analysis of their algorithm relies upon the use of a boosted and shift-invariant version of QPE
which may not exist \cite{chen2023quantum}.
Recent efforts to design  quantum Gibbs samplers  take a very different approach and are based on simulating Davies generators ~\cite{chen2023quantum,chen2023efficient,rall2023thermal,ding2024efficient}.
Currently, these are the only provably correct Gibbs samplers for quantum Hamiltonians.

We revisit the inspiration for the  Metropolis-style  algorithm of~\cite{temme2011quantum} and incorporate weak measurement to design a conceptually simple and provably correct quantum Gibbs sampler,  with the Gibbs state as its approximate unique fixed point. 
    Our method uses a Boosted QPE which takes the median of multiple runs of QPE, but we do not require the shift-invariant property. In addition, we do not use the Marriott-Watrous rewinding technique which  simplifies the algorithm significantly.

\end{abstract}
\newpage 
\tableofcontents

\newpage

\section{Introduction}

One of the primary motivations for building quantum computers is to simulate quantum many-body systems.
While there has been significant progress in simulating quantum  dynamics~\cite{berry2014exponential,berry2015hamiltonian,cleve2016efficient}, much less is known about  preparing
 ground states and Gibbs states, an essential task in  understanding the static properties of a system. 
 In particular, the properties of Gibbs states, which describe the thermal equilibrium of a system at finite temperature, are closely related to central topics in 
 condensed matter physics and quantum chemistry~\cite{alhambra2023quantum}, such as  the electronic binding energy and phase diagrams.
 In addition to applications in physics, Gibbs states are widely used as generative machine learning models, such as  classical and  quantum Boltzmann machines~\cite{hinton1983optimal,amin2018quantum}.
 Algorithms for preparing Gibbs states are also used as   subroutines  in other applications, such as solving semidefinite programs~\cite{brandao2019quantum,van2017quantum}.

 Typically, a good Gibbs state preparation algorithm (Gibbs sampler) should satisfy two requirements:
 it should have the Gibbs state as its \textit{(unique) fixed point} and it should be \textit{rapidly mixing}. 
 The  fixed point property ensures the \textit{correctness} of the Gibbs sampler.
An algorithm that keeps the Gibbs states invariant and shrinks any other state will eventually converge to the Gibbs states after  a sufficiently long  time. 
 Our algorithm satisfies an approximate version of the fixed point property.
 The  mixing  time determines the \textit{efficiency} of the algorithm. In particular, an algorithm is said to be fast mixing  if it   convergences to  Gibbs state  in $poly(n)$ time. 
 
 The focus of this manuscript is the correctness part, that is, designing a  Gibbs sampler which satisfies  the  fixed point property. 
 For classical Hamiltonians, like the Ising model, the  fixed point property is easily satisfied by the celebrated Metropolis algorithm~\cite{metropolis1953equation} which has become one of the most widely used algorithms throughout science. For quantum Hamiltonians, designing an algorithm which provably converges to the Gibbs state has been more challenging. As noticed in the pioneering work of ~\cite{temme2011quantum} ten years ago, designing a Metropolis-type algorithm for quantum Hamiltonians is non-trivial,  mainly due to two reasons: 

 \begin{itemize}
     \item[(1)] Quantum computers with finite resources cannot distinguish eigenvalues and eigenstates with infinite precision. 
     \item[(2)]  Mimicking the rejection process in quantum Metropolis requires reverting a quantum measurement.
 \end{itemize}

  \cite{temme2011quantum} eased the first challenge using a boosted version of QPE  which sharpens  the accuracy by taking the  median of multiple runs. They addressed the second challenge using the Marriott-Watrous rewinding technique along with a shift-invariant version of QPE in the case of a rejected move. 
  As a result, their analysis  depends on  a version of $\sQPE$
  that  is both boosted and shift-invariant. 
   Recent work~\cite{chen2023quantum} suggests  a version of $\sQPE$ with both of those properties may  be impossible. While \cite{temme2011quantum} provides many innovative ideas, a provably correct quantum Gibbs sampler remained elusive for some time. 
   Recently, Chen \textit{et.al.}~\cite{chen2023quantum} designed the first Gibbs sampler which provably satisfies the fixed point property for general Hamiltonians  based on a significantly different approach.  Their method is based on simulating   quantum master equations (Lindbladians)
  which more closely mimics the way that systems thermalize in nature. Their algorithm approximately simulates the Davies generator~\cite{davies1976quantum,davies1979generators}, which describes the evolution of quantum systems coupled to a large heat bath in the weak coupling limit. It is worth mentioning that the Davies generator by itself also assumes the ability to distinguish  eigenvalues with infinite precision, and thus cannot be efficiently simulated  by quantum computers.  To resolve this problem, \cite{chen2023quantum} devised a method to smooth the Davies’ generator by using a weighted operator Fourier Transform for Lindbladians.

 Although \cite{chen2023quantum} provides a provably correct quantum Gibbs sampler by  simulating the thermalization process occurring in  nature, it is natural to ask whether a Metropolis-style quantum Gibbs sampler can be designed. 
 Are there intrinsic reasons why  an algorithm  based on the classical Metropolis  process cannot work for quantum Hamiltonians? 
Or on the other hand,
\begin{quote}
   Is it possible to design a provably correct quantum  Gibbs sampler, which is analogous to the conceptually  simple    classical Metropolis  algorithm?
\end{quote}
\noindent
 In this manuscript, we give an affirmative answer to the above question, by designing a simple Metropolis-style Gibbs sampler. 
  Our algorithm uses many of the components of 
 \cite{temme2011quantum}, but there are some key differences:
 (1) We  use \textit{weak measurement} in determining whether to accept or reject a given move. (2) 
  We use a Boosted QPE and do not assume the shift-invariant property. (3) We do not use Marriott-Watrous Rewinding tehchnique~\cite{marriott2005quantum}, which  simplifies the algorithm considerably.

For (3), more precisely,
  recall that  the  mechanism that \cite{temme2011quantum} uses to back up in a reject case 
 requires a $poly(n)$ sequence of  forward and backward unitaries and complex measurements until the backing up process succeeds. 
 In comparison, 
 we   use only \textit{one single-qubit measurement} and \textit{one  unitary} for rewinding.  This simplification is achieved by noticing that   after one round of unsuccessful rewinding, the  state is almost equivalent to the 
 state in the accept case. This allows us to accept and conclude the iteration in one step instead of attempting to rewind again. We call this case an {\it Alternate Accept}. 
 This observation  is an essential feature of our analysis, since  we would still need to perform Marriott-Watrous rewinding  without the Alternate Accept case, even with the weak measurement.
One  remark is that while
weak measurement helps in rewinding,  it comes  at the expense of increasing  the number of iterations by a polynomial factor. Also, while our algorithm is simple, it does not have the most favorable scaling as a function of the system parameters and desired precision.

We note that weak measurement is also one of the reasons   why the  approaches based on the Davies generator  succeed, since simulating a Lindbladian requires the use of weak measurement.  Our algorithm also effectively approximates the evolution of a Lindbladian. In this sense our algorithm is conceptually similar to \cite{chen2023efficient}. The key difference is that our Gibbs sampler is directly designed from QPE instead of a Davies generator.

\section{Overview}

\subsection{Algorithm Overview}
\label{sec:algover}
Given a local Hamiltonian $H$ and an inverse temperature $\beta$, the goal is to design a quantum algorithm which prepares the Gibbs states
$\rho_\beta= \exp(-\beta H)/Z$, where $Z=tr(\exp(-\beta H))$.

Our quantum  algorithm attempts to mimic the  classical Metropolis algorithm  similar to \cite{temme2011quantum}. Recall that the classical Metropolis algorithm is a random walk whose states are eigenstates of a classical Hamiltonian. In each iteration, the algorithm starts in some state $x$ with energy $\nu_x$. A jump operator is applied to alter $x$ in some way to obtain a new state $y$ with energy $\nu_y$. Then a randomized decision is made whether to  \mbox{\acc} the move and remain in state $y$, or  \mbox{\rej}  the move and revert back to $x$. The acceptance probability is defined by a function $f_{\nu_x\nu_y}\in [0,1].$   The Metropolis acceptance rule is designed so that the random walk converges to the Gibbs state. In particular, the rule must satisfy  
$$		\exp(-\beta \nu_x )f_{\nu_x\nu_y}= \exp(-\beta \nu_y )f_{\nu_y\nu_x}.$$

	Metropolis sampling uses the following function $f$:
$$		f_{\nu_x\nu_y}:=\min \left\{1, \exp\left(\beta \nu_x-\beta\nu_y\right) \right\}.$$

 The main obstacle in adapting the classical Metropolis algorithm to the quantum setting, is that a measurement must be performed in deciding whether to accept or reject. Then   
 in the  \mbox{\rej} case, one needs to rewind back to the  state 
  before the move, thus reverting a quantum measurement, which is  hard. The algorithm presented here manages this difficulty effectively by using weak measurement in determining whether to \mbox{\acc} or \mbox{\rej}. It is worth noting that, in contrast to the classical case where we can compute the energy $\nu_x$ exactly, there are intrinsic limitations on our ability for  estimating energies  of quantum states (due to the energy-time uncertainty principle). Analyzing  the errors incurred by imperfect quantum energy estimation is non-trivial, and is one of the most technical parts in  all related works~\cite{chen2023quantum,chen2023efficient,wocjan2023szegedy,temme2011quantum}. We will explain more in the overview of techniques section.

The quantum algorithm uses four registers. The first is an $n$-qubit register which stores the current state of the algorithm. The next two each have $gr$ qubits for integers $g,r$ and are used to store the output of an application of the Boosted Quantum Phase Estimation ($\QPE$) algorithm, which provides an estimate of the state's eigenvalue.  The last register is a single qubit register which controls whether we accept or reject the new state.

The algorithm uses three different operations outlined below.
Let $H$ be an $n$-qubit local Hamiltonian.  We  use $\left\{\ket{\psi_j},E_j\right\}_j$ to denote an ortho-normal eigenbasis of $H$ and their corresponding eigenvalues. 

\begin{description}
    \item{\bf Boosted Quantum Phase Estimation ($\QPE$):}  
    $\QPE$ is a unitary on two registers  of $n$ and $gr$ qubits respectively.   If $\QPE$ starts with an eigenstate of $H$ in the first register and the  second register is initialized to $g$ copies of $\ket{0^r}$, then $\QPE$ corresponds to performing $g$ independent iterations of the standard Quantum Phase Estimation with respect to the first register and storing the result in each copy of $\ket{0^r}$. This process leaves the first register unchanged and outputs $g$ independent estimates of $E_j$ in the second register. Each $r$-bit string $\bb$ in the second register represents an energy $E(\bb)$ defined as:
    \begin{align}
 &E(\bb):= 	\kappa_H \sum_{j=1}^r b_j2^{-j}\label{eq:Ebb_main}.
 \end{align}
 where $\kappa_H$  is a power of two that upper bounds $\|H\|$.
The set $S(r):=\{E(\bb)\}_{\bb\in\{0,1\}^r}$ is  the set of energies that can be represented by $r$-bit strings, which are integer multiples of $\kappa_H \cdot 2^{-r}$.  To ease notation, we use notation $\ket{E}$ for $E \in S(r)^{\otimes g}$ as the  basis of  the second register, instead of using   strings in $\{0,1\}^{rg}$. Thus, $\QPE$ operates as:    
    \begin{align}
        \QPE	\ket{\psi_j}\ket{0^{gr}}= \ket{\psi_j} \sum_{E \in S(r)^{\otimes g}} \beta_{jE} \ket{E}\label{eq:overview_QPE}
    \end{align}
    The cost of $\QPE$ is  $g\cdot poly(2^r,n)$.

 We write $\overline{E}$ as the median of the $g$ energy estimates in $E$.  We denote  $\flo{E_j}, \cei{E_j}$ as   the  best two approximations of $E_j$ in $S(r)$, that is the 
 closet  value to $E_j$ which is an integer multiples of $\kappa_H \cdot 2^{-r}$ and is smaller/greater than $E_j$ respectively. 
More details on $\QPE$ are given in Appendix \ref{appendix:QPE}. 
 
 We also use a  variant of $\QPE$, which we call $\FQPE$ (Flipped Boosted Quantum Phase Estimation), where  the amplitudes of the output of $\FQPE$ are the complex conjugates of $\QPE$:
 $$\FQPE	\ket{\psi_j}\ket{0^{gr}} = \ket{\psi_j} \sum_{E \in S(r)^{\otimes g}} \beta^*_{jE} \ket{E}.$$
 The implementation of $\FQPE$ is a slight modification of  $\QPE$ and  is given in Appendix \ref{appendix:QPE}. It is worth noting that $\FQPE \neq \QPE^\dagger$.
    
    \item{\bf Jump operators:}
    A set of unitaries $\{C_j\}_j$ called {\it jump operators} and a distribution $\mu$ over this set.
    We require that the set $\{C_j\}_j$ is closed under adjoint. In addition, for any $C \in \{C_j\}_j$,  we require that $\mu$ chooses $C$ and $C^\dagger$ with the same probability. 
    We use $C \leftarrow \mu$ to denote a selection of $C$ drawn according to distribution $\mu$. 
    
    To ensure the uniqueness of the fixed point, we also require that the algebra generated by $\{C_j\}_j$ is equal to the full algebra, that is the set of all $n$-qubit operators. For example one can choose $\{C_j\}_j$ to be all single-qubit Paulis.
    \item{\bf Acceptance Operator ($W$):} Finally, we use a unitary which calculates the acceptance probability based on the two energies stored in registers $2$ and $3$, scaled by a factor of $\tau^2$ and rotates the last qubit by the square root of the acceptance probability. More precisely,    $W$ operates on registers $2$, $3$, and $4$ as:
	\begin{align}
		&W := \sum_{E,E'\in S(r)^{\otimes g}} \ket{EE'}\bra{EE'}\otimes 	\begin{bmatrix}
 			\sqrt{1-\tau^2  f_{EE'}}&  \tau  \sqrt{f_{EE'}}\\
 					\tau  \sqrt{f_{EE'}} & -\sqrt{1-\tau^2 f_{EE'}}
 					\end{bmatrix}.
      \end{align}
   
where $f_{EE'}$ is the Metropolis acceptance rate based on the median of energy estimates:
\begin{align} f_{EE'}:=\min \left\{1, \exp\left(\beta \overline{E}- \beta \overline{E'}\right) \right\}	
\end{align}
 We can think of $W$ as computing the median of $E$ and $E'$ to get $\overline{E}$ and $\overline{E}'$ respectively and then rotating the last qubit by $f_{EE'}$.
 The median operation $\overline{E}$ is used to boost the energy estimation, suppressing the probability of an incorrect estimate with Chernoff bounds. 
     The operator $W$ is the same as the one used in \cite{temme2011quantum} with the addition of the slow-down factor of $\tau^2$. Note that
      \begin{align}
      W \ket{EE'}\ket{0} =   \ket{EE'} \left( \sqrt{1-  \tau^2 f_{EE'}}\ket{0}+  \tau \sqrt{f_{EE'}}	\ket{1}\right)
	\end{align}
\end{description}

With those components defined, we can now describe an iteration of our algorithm depicted in Figure \ref{fig:iter}. 
The Algorithm Sketch below is a  high-level overview of the algorithm. The complete pseudo-code is given in Section \ref{sec:algo}.

\begin{figure}[h!] \centering
	\includegraphics[width=0.8\textwidth]{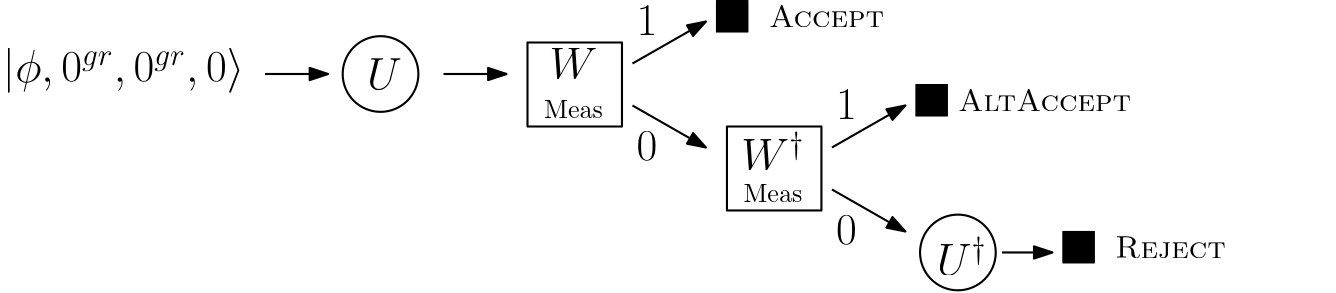}\caption{One iteration of the algorithm. The operation $U$ is $\sQPE_{1,3} \circ C \circ \sQPE_{1,2}$. The two measurements are performed on the last qubit only. The $\blacksquare$ symbol indicates that the last three registers are traced out and replaced by fresh qubits in the $\ket{0}$ state. } \label{fig:iter} 
\end{figure}

{\bf Algorithm Sketch:}  In each iteration, the algorithm chooses one of $\{\QPE,\FQPE\}$ with equal probability. The selected operation is called QPE. 
A jump operator $C \leftarrow \mu$ is also selected.
We use $\cC$ to denote the random choices for QPE and $C$ made in a particular iteration. 
The algorithm starts an iteration with a state $\ket{\phi}$ in register $1$ and all $0$'s in the other three registers. Then: 
 
\begin{itemize}
\item QPE is first applied to the current state, and the estimate of the eigenvalue is stored in register $2$. 
 Then the jump operator $C$ is applied to the state in register $1$ to obtain a new state.  QPE is then applied to the new state in register $1$ and  the estimate of its eigenvalue is stored in register $3$. We call the sequence of these three operations $U_{\cC} := \sQPE_{1,3} \circ C \circ \sQPE_{1,2}$.

\item Then $W$ is applied to registers $2$, $3$, and $4$, and the last qubit is measured to get {\sc Outcome1}.
    \item If {\sc Outcome1}  $= 1$, the algorithm accepts  the move (Case: \acc) and continues. 
    \item If {\sc Outcome1}  $= 0$, then $W^{\dag}$ is applied and the  last qubit is measured again to get {\sc Outcome2}.
    \begin{itemize}
    \item The case in which {\sc Outcome2} $ = 1$  represents an alternative way of accepting the move (Case: \aacc).
    \item If {\sc Outcome2} is $0$, then this represents a rejection of the move 
(Case: \rej), in which case  $U_\cC^{\dagger}$ is applied.
\end{itemize}
\item Finally, registers $2$, $3$, and $4$ are traced out and replaced by fresh qubits in  all $0$'s states. 
\end{itemize}

\subsection{Informal Statement of Results}

Let $\cEN(\rho)$ be the map corresponding to one iteration of the algorithm.
Recall that  $\tau$ is the parameter in the acceptance operator $W$, and $g,r$ are the precision parameters in $\QPE$. Our main result is proving that
 our algorithm approximately fixes the Gibbs states:
\begin{theorem}[Informal version of theorem \ref{thm:main}] \label{thm:intro_main} $\cEN$ can be expressed as
$$\cEN  
= \cI + \tau^2 \cLt + \tau^4 \cJt.$$
where $\cLt$ is independent of $\tau$ and approximately fixes the Gibbs state. More precisely
for any $\delta$, parameters $g$ and $r$ can be chosen so that  $g=O(n+\log 1/\delta)$ and $r= O(\log \beta +\log \kappa_H+\log 1/\delta)$, and
$$\tn{\cLt(\rho_\beta)}\leq \delta.$$ 
\end{theorem}
 Here $\tn{\cdot}$ refers to the trace norm.

 The proof that the fixed point of our algorithm is approximately the Gibbs State depends upon the assumption that
$\cLt$ is fast mixing, meaning that $\cLt$ converges to its fixed point $\rho_\cLt$ in $poly(n)$ time. More precisely,
we can combine the fast mixing property with the fact that $\tn{\cLt(\rho_\beta)}\approx 0$ from Theorem \ref{thm:intro_main} above to establish that  the fixed point of $\cLt$  is close to the true Gibbs state: $\rho_\cLt\approx \rho_\beta$. 
The next step then is to show that starting from an arbitrary state, repeated iterations of our algorithm   will result in a good approximation of $\rho_\cLt$. In particular,
the second part of Theorem \ref{thm:intro_main}, which says that 
 $\cEN=\cI+\tau^2 \cLt+\tau^4 \cJt$, implies that $\cEN^K\approx e^{K\tau^2 \cLt}$, where $\cEN^K$ corresponds to $K$ iterations of the algorithm. The error in the approximation scales as $K\tau^4$.
$K$ and $\tau$ can be chosen so that  $K\tau^2$ is polynomial in $n$ but the approximation error $K \tau^4$ is still small.  Assuming that $\cLt$ is fast mixing, we can conclude that $ e^{K\tau^2 \cLt}$  drives any states to $\rho_\cLt$. This reasoning leads to Theorem \ref{thm:intro_mix} below. 
\begin{theorem}[Informal version of Theorem \ref{thm:mix}]\label{thm:intro_mix}
    Suppose the mixing time of $\cLt$ is $poly(n)$. Assume $\beta,\kappa_H\leq poly(n)$.  For parameters  $\tau=1/poly(n)$, $g=O(n)$, $r=O(\log n)$,  $K=poly(n)$, and  for any initial state $\rho$, we have
    \begin{align}
        \tn{\cEN^K(\rho)-\rho_\beta}\leq 1/poly(n).
    \end{align}
\end{theorem}
 Theorem \ref{thm:intro_mix} can also be rephrased in terms of the spectral gap  similarly as in the classical Metropolis algorithm. The formal statement is in Corollary \ref{cor:eb_spe_gap}. 

For the above choice of parameters the cost of implementing   $\cEN^K$
is $Kg\cdot poly(2^r,n)=poly(n)$, where the $poly(2^r,n)$ is mainly the cost of implementing one run of standard Quantum Phase Estimation. 

\subsection{Overview of Techniques}

Intuitively our algorithm  should approximately fix the Gibbs state, since it mimics the classical Metropolis. 
The approximation errors come from two resources, one is controlled  by the parameter $\tau$ in the acceptance operator $W$, and the other is controlled by the  $g$ and $r$ in $\QPE$.

Let us begin with the error from $\tau$. According to the algorithm, 
  $\cEN$ can be expressed as the sum of three operators representing the three cases: $\cEN^{(a)}$ (for \acc), $\cEN^{(b)}$ for (\aacc), and $\cEN^{(r)}$ (for \rej).
 So that $\cEN = \cEN^{(a)} + \cEN^{(b)} + \cEN^{(r)}$.
 Each of these can be further expanded to represent their dependence on $\tau$. 
 Mores specifically, in Section \ref{sec:CTMC}, we define additional operators, $\cM^{(a)}$, $\cM^{(r)}$, $\cJt^{(b)}$, $\cJt^{(r)}$, all with norm  bounded by a constant. We show that
\begin{align}
 \mbox{\acc:} ~~~  \cEN^{(a)} & = \tau^2 \cM^{(a)}\\
  \mbox{\aacc:} ~~~   \cEN^{(b)} & = \tau^2\cM^{(a)} + \tau^4 \cJ^{(b)}\\
  \mbox{\rej:} ~~~   \cEN^{(r)} & = \cI - \tau^2\cM^{(r)} + \tau^4 \cJ^{(r)}
\end{align}
Note that the operators for the \acc~ and \aacc~ cases only differ by an operator on the order of $\tau^4$, which means that the state in the \aacc~  is very close to the resulting state in the \acc~ case. 
By defining $\cLt = 2 \cM^{(a)} - \cM^{(r)}$ and $\cJ = \cJ^{(b)} + \cJ^{(r)}$, we have that
$$\cEN 
= \cI + \tau^2 \cLt + \tau^4 \cJt.$$

Most of the technical effort in the proof of Theorem \ref{thm:intro_main} is spent  showing that $\tn{\cLt(\rho_\beta)}\approx 0$. There are two features of $\QPE$ that make this fact challenging to prove. The first feature is that $\QPE$ cannot be made deterministic. More precisely, recall that  $\flo{E_j}$ and $\cei{E_j}$ are  the  best two approximations of $E_j$.  Without additional assumptions on the Hamiltonian,  the amplitudes  $\beta_{j\flo{E_j}}$ and $\beta_{j\cei{E_j}}$ are both non-negligible. 
We use  the boosted version of $\QPE$ which guarantees that the probability of generating an estimate $E$ that is different from $\flo{E_j}$ or $\cei{E_j}$ is negligibly small. However, the fact that the output of $\QPE$ will still be a superposition of 
$\flo{E_j}$ and $\cei{E_j}$ is unavoidable and makes the process inherently different from a classical random walk. Mathematically, this means that if  the algorithm begins in an eigenstate $\ket{\psi_j}$, after one iteration, the new state is no longer diagonal in the energy eigenbasis.
In particular,  $\cLt(\rho_\beta)$ can have  exponentially many non-zero off-diagonal entries when expressed in the energy eigenbasis.

The second feature of $\QPE$ that makes the analysis problematic  is that there are intrinsic limitations on the precision of energy estimation of quantum states. In particular, the energy/time uncertainty principle suggests that a $poly(n)$-time quantum algorithm (like $\QPE$) can only estimate the energy of a state within $1/poly(n)$ precision. This means that the off-diagonal entries of $\cLt(\rho_\beta)$ can  have magnitude on the order of $1/poly(n)$.
The fact that $\cLt(\rho_\beta)$ can  have exponentially many off-diagonal entries that have magnitude $1/poly(n)$ rules out a brute-force approach to bounding
$\tn{\cLt(\rho_\beta)}$.

To illustrate our approach in overcoming these technical difficulties, first imagine instead an ideal version of $\QPE$ which deterministically maps 
every $\ket{\psi_j}$ to $\flo{E_j}$.
 The resulting process has a fixed point which is a Gibbs state where the probability of each
state is proportional to $e^{-\beta \flo{E_j}}$ instead of $e^{-\beta E_j}$. 
Call this truncated fixed point $\rho_{\beta0}$.  
It's not too hard to show that $\rho_{\beta0} \approx \rho_{\beta}$, which implies $\cLt(\rho_{\beta0})\approx \cLt(\rho_{\beta})$, so we will focus instead on 
bounding $\tn{\cLt(\rho_{\beta0})}$.
For the practically realizable, non-deterministic $\QPE$,
our analysis effectively decomposes $\cLt(\rho_{\beta0})$ into a sum of a constant number of terms and bounds the trace norm of each term separately by a $1/poly(n)$.
The different terms are derived by inserting different projectors that separate out
the cases when $\QPE$ maps a state $\ket{\psi_j}$ to $\flo{E_j}$, $\cei{E_j}$, or some other $E$ entirely. We can represent the cases by defining projectors:
\begin{align*}
    P^{(0)} & = \sum_j \ketbra{\psi_j}{\psi_j} \otimes \sum_{E:\overline{E}=\flo{E_j}} \ketbra{E}{E}\\
    P^{(1)} & = \sum_j \ketbra{\psi_j}{\psi_j} \otimes \sum_{E:\overline{E}=\cei{E_j}} \ketbra{E}{E}.
\end{align*}
When the matrix
$$\cLt(\rho_{\beta0}) = \sum_j \frac{e^{-\beta \flo{E_j}}}{Z} \bra{\psi_m} \cLt( \ket{\psi_j} \bra{\psi_j} ) \ket{\psi_n}$$
is written out, there are four applications of phase estimation: $\QPE_{1,2}$ and $\QPE_{1,3}$ are applied to $\ket{\psi_j}$, and $\QPE^{\dag}_{1,2}$ and $\QPE^{\dag}_{1,3}$ are applied to $\bra{\psi_j}$. This results in 
 a total of $16$ terms 
depending on which of the two projectors ($P^{(0)}$ or $P^{(1)}$) is applied after each occurrence of $\QPE$. Recall that $\cLt=2\cM^{(a)}-\cM^{(r)}$. 
In most cases, we don't get exact cancellation between the accept and reject   operators because each case may have some multiplicative error of the form $e^{\pm \beta \delta}$, where $\delta = \cei{E_j} - \flo{E_j}$ is the precision of $\QPE$ (which is independent of $j$). 
The essential observation is that each separate case results in exactly the same multiplicative error. This is because in each case, all of the $\QPE$ estimates are erring in exactly the same direction and by exactly the same amount. We get that $\cLt(\rho_{\beta0})$ can be expressed as a linear combination of terms $\{N_j\}$:
\begin{align}
\label{eq:cases}
    \cLt(\rho_{\beta0})\approx \sum_{j=1}^{16} \mbox{error}_j \cdot N_j,
\end{align}
where $\mbox{error}_j = e^{\pm \beta \delta} -1$ and $\tn{N_j} = O(1)$. In the proof we manage to reduce the number of terms from $16$ to $6$ by carefully clustering terms.
Note that the parameters are chosen so that $e^{ \beta \delta} = (1 + 1/\mbox{poly}(n))$.
Equation (\ref{eq:cases}) is still approximate because we have not yet taken into account the case where
 $\QPE$  maps $\ket{\psi_j}$ to some value other than $\flo{E_j}$ or $\cei{E_j}$. For this case, we have a third projector:
$$P^{(else)} = \sum_j \ketbra{\psi_j}{\psi_j} \otimes \sum_{E: \overline{E}  \neq \flo{E_j},\cei{E_j}} 
 \ketbra{E}{E}.$$
When $P^{(else)}$ is used, the norm of the resulting operator is exponentially small
because of the use of boosted QPE. Specifically, we use the following lemma which is 
included in Appendix \ref{appendix:QPE}. 
  \begin{lemma}
\label{lem:QPEerr}
    If $\ket{\psi_j}$ is an eigenstate of $H$ with energy $E_j$, then 
\begin{align}
	&\QPE	\ket{\psi_j}\ket{0^{gr}} = \ket{\psi_j} \sum_{E \in S(r)^{\otimes g}} \beta_{jE} \ket{E},\label{eq:QPE}\\
	\text{where }& \sum_{E\in S(r)^{\otimes g}:\,\overline{E}\neq \flo{E_j},\overline{E}\neq \cei{E_j}}  |\beta_{jE}|^2  \leq 2^{-\cg }.\label{eq:peak}	
	\end{align}
\end{lemma}

Our proof works for general Hamiltonians and we do not assume non-physical assumptions like rounding promise~\cite{wocjan2023szegedy}.
One more remark is that $\QPE$ itself will create some biased phase in the resulting states. In our algorithm, we choose $\QPE$ and $\FQPE$ randomly to cancel this bias. 
Finally, the proof of uniqueness of the fixed point, which appears in Section \ref{sec:uniqueness}, is standard, and is based on showing that $\cLt$ is of Lindbladian form, and that the generators of the Lindbladian generate the full algebra of operators on the $n$-qubit Hilbert space.

\subsection{Related Work}

We have  recently become aware of the concurrent, independent work of \cite{GCDK24}, which also provides a quantum generalization of Glauber/Metropolis dynamics. In contrast with our algorithm, their method does not use QPE.
Instead, they construct a quantum  extension of discrete and continuous-time Glauber/Metropolis dynamics, 
in the style of a smoothed Davies generator~\cite{chen2023quantum, chen2023efficient}.  They show that their construction exactly fixes the Gibbs states and can be  efficiently implemented on a quantum computer. For the continuous-time case, their implementation is achieved by   Linbladian simulation~\cite{chen2023quantum}.
For the discrete-time case, they use
 oblivious amplitude amplification~\cite{gilyen2019quantum} in combination with techniques based on linear combinations of unitaries and the quantum singular value transform.

Before the appearance of \cite{chen2023quantum}, 
there were many previous results  on simulating Davies generators, with additional assumptions on the Hamiltonian. As mentioned before,
Davis generators cannot be implemented exactly since they require the ability to estimate  eigenvalues perfectly, which is impossible with  quantum computers with finite resources. \cite{wocjan2023szegedy} circumvents this problem by assuming a rounding promise on the Hamiltonian, which 
disallows eigenvalues from certain sub-intervals. \cite{rall2023thermal} later eliminates the  rounding promise by using randomized rounding, which at the same time incurs  an additional resource overhead.   
More recently, \cite{chen2023efficient} designed a weighted version of Davies generator which exactly fixes the Gibbs states (before truncating the infinite integral to a finite region).  \cite{ding2024efficient} 
introduces a family of quantum Gibbs samplers satisfying the Kubo-Martin-Schwinger detailed balance condition,  which includes the construction of  \cite{chen2023efficient}  as a special instance.
In addition to approaches based on the Davies generator, 
there are Grover-based approaches~\cite{poulin2009sampling,chowdhury2016quantum},  which prepare a purified version of Gibbs states. The performance of those algorithms depends on the overlap between the initial and the target state. 

All of the above approaches use either quantum phase estimation or quantum simulation, which involve large quantum circuits. 
In contrast to those approaches,
\cite{zhang2023dissipative} designs a dissipative quantum Gibbs sampler with  simple local
update rules. \cite{zhang2023dissipative} differs from the ordinary  Gibbs samplers in that the Gibbs state is not generated as the fixed point of the Markov Chain, but is generated on a conditionally stopped process.

 All the above work is focused on satisfying the fixed point requirement.
 A different but important task is to bound the mixing time, which  is  wide open with the exception of a few special cases. In particular,  \cite{kastoryano2016quantum} shows that for a commuting Hamiltonian on a lattice, there is an equivalence  between very rapid mixing (more precisely, constant spectral gap of the Linbladian)  and a certain strong form of clustering of correlations in the Gibbs state, which generalizes the classical result~\cite{stroock1992equivalence,stroock1992logarithmic} to the quantum setting. 
 \cite{bardet2023rapid} proves fast mixing for 1D commuting Hamiltonian at any temperature. Recently \cite{rouze2024efficient,bakshi2024high} prove fast mixing for non-commuting Hamiltonian at high temperature.

There are also many heuristic methods for preparing  Gibbs states. Previous proposals include methods that emulate the physical thermalization process by repeatedly coupling the systems to a thermal bath~\cite{terhal2000problem,shabani2016artificial,metcalf2020engineered}.  
There are also approaches based on quantum imaginary time evolution~\cite{motta2020determining} and variational algorithms~\cite{wang2021variational,lee2022variational,consiglio2023variational}.

\subsection{Conclusions and Future Work}

In this manuscript, we use weak measurement to design a quantum  Gibbs sampler based on the Metropolis algorithm which satisfies the approximate  fixed point property. Compared with previous work, the main advantage of our algorithm is its conceptual simplicity.
Our algorithm  uses a Boosted QPE which takes the median value of multiple independent runs of the standard QPE. We do not require a version of QPE that satisfies shift-invariance. In addition,  our algorithm is free of 
Marriott-Watrous rewinding techniques and
only performs  single-qubit measurements.  Technically, we  give a new way of bounding the errors incurred by imperfect energy estimation of QPE, by grouping the error terms into finite classes. Our proof works for general Hamiltonians and we do not assume the  rounding promise.
It might be interesting to investigate whether this technique can be applied to prove that the existing Gibbs Sampler for Hamiltonians with rounding promise~\cite{wocjan2023szegedy} in fact works for general Hamiltonians.

While our algorithm is simple, it comes at a cost of
 not having the most favorable scaling in the  parameters of the system and desired precision. In particular, our algorithm effectively corresponds to directly simulating 
 a Lindblad evolution  $e^{t\cLt}$
 by discretizing $t$, where the cost is scaled as $O(t^2/\epsilon)$ for precision parameter $\epsilon$. 
It is worth noting, however, that  our algorithm itself is designed directly from QPE, and  we use $e^{t\cLt }$ only in our analysis,
as opposed to first designing the  Lindbldian $\cLt$ and attempting to simulate it on a quantum computer.
It would be interesting to explore how more sophisticated Linbladian simulation techniques (such as ~\cite{li2022simulating,cleve2016efficient,childs2016efficient}) could be applied to our algorithmic structure to improve the dependence on parameters. Note that our $e^{t\cLt}$  circumvents the problem of precision in the Davies generator, because the operator $\cLt$ is already defined in terms of QPE of finite precision.

Another possible direction for future work is to compare the mixing time of different Gibbs samplers. In particular, it would  be interesting to compare, either numerically or theoretically, the spectral gap of the Davies generator, the Davies-generator-inspired Lindbldian in \cite{chen2023quantum,chen2023efficient}, and the Lindbladian $\cLt$ in our algorithm. Since our algorithm is similar to the classical Metropolis algorithm, it might also be interesting to explore whether classical techniques for analyzing mixing times can be generalized to the quantum setting.

\subsection{Structure of the manuscript}

The manuscript is structured as follows. In 
Section \ref{sec:pre} we give necessary  definitions and notations. In Section \ref{sec:algo} we give explicit pseudo-code of our  algorithm and express the result of each operation more formally. 
We  state our main results in Theorem \ref{thm:main} and Theorem \ref{thm:mix} at the beginning of Section \ref{sec:CTMC}.

The proofs of  Theorem \ref{thm:main} and Theorem \ref{thm:mix} are divided into several sections. Section \ref{sec:CTMC}  contains the derivation of the explicit formula for $\cLt$. In Section \ref{sec:Lind_unique} we show that $\cLt$ can be written as a Lindbladian and  has a unique fixed point. Section \ref{sec:Lind_unique} is independent of Section \ref{sec:fp} and can be skipped temporarily.
In Section \ref{sec:fp} we prove  $\cLt$ approximately fixes the Gibbs state.
  Finally in Section \ref{sec:eb_final} we prove  Theorem \ref{thm:main} and Theorem \ref{thm:mix}.

\section{Preliminary and Settings}\label{sec:pre}

\subsection{Gibbs states and Assumptions}\label{sec:pre_Gibbs}

For any $n$-qubit  Hamiltonian $H$, 
we always use $\left\{\ket{\psi_j},E_j\right\}_j$ to denote an ortho-normal set of (eigenstate, eigenvalue) for $H$.  We use  symbols different from $\ket{\psi_j}$ to denote other states.
For any
  inverse temperature $\beta\geq 0$,  we denote the Gibbs state as 
\begin{align}
   \rho_\beta(H) &:= \exp(-\beta H) / tr(\exp(-\beta H))\\
   &= \sum_j p_j \ket{\psi_j}\bra{\psi_j},\\
  \text{where } p_j &:= \exp(-\beta E_j)/ tr(\exp(-\beta H)).\label{eq:def_p}
\end{align}
We assume $ H\geq 0$, and its spectrum norm is bounded by $\|H\|\leq poly(n)$. Note that one can always add multiples of identity matrices to $H$ to ensure $H\geq 0$ and this operation does not change the Gibbs state.
To ease notation, we will fix $H$ and abbreviate $\rho_{\beta}(H)$ as $\rho_{\beta}$.

$\kappa_H=poly(n)$ is a power of two that upper bounds  
$\|H\|$.  For example, for local Hamiltonian $H=\sum_{i=1}^{m} H_i, \|H_i\|\leq 1$, one can set $\kappa_H$ to be the least integer which is a power of two and is greater than $m$. Then $\kappa_H\leq 2m$.

	For simplicity,
in this manuscript we assume that we can implement arbitrary 2-qubit gates exactly. Note that this assumption does not influence the generality of our results, since 
 the error analysis can be easily generalized to the practical case where we  approximate arbitrary 2-qubit gate to $1/poly(n)$ precision, by noticing that the $l_2$ norm $\|(U-V)\ket{\psi}\|_2$ is bounded by the spectrum norm $\|U-V\|$ for any $\ket{\psi}$.

\subsection{Notations and Norms.}
We use $\log $ for $\log_2$.
For a complex value $a\in \bC$, we use $a^*$ to represent its complex conjugate. For two numbers $x,y$, we use $\delta_{xy}$ to denote the function which equals to $1$ if $x=y$ and $0$ otherwise.
For a matrix $M$, we use $M^\dagger$ to denote its complex 
conjugate transpose. 
For two matrices $M,N$, 
we use $\{M,N\}_+ $ to denote their anti-commutator: $MN+NM$.
We use
 $\|M\|$ to denote the spectrum norm of $M$. For a vector $\ket{\phi}$ we use $\|\ket{\phi}\|_2$ to denote the $l_2$ norm.
$\ket{\phi}$ can be normalized or unnormalized.
 When it is necessary, we will use number subscripts to denote the name of the quantum registers. For example,  $\ket{\phi}_1$ means the state is in register $1$.

We use $\Xi(m)$ to denote the set of linear operators on an $m$-qubit Hilbert space. We use $\cH(m)$ to denote the set of \textit{Hermitian} linear operators on an $m$-qubit Hilbert space.  We say $\rho\in\cH(m)$ is an $m$-qubit quantum state if $\rho\geq 0,\rho=\rho^\dagger$ and $tr(\rho)=1$.
We use $I_m$ to denote the identity matrix on $m$ qubits. When $m$ is clear we abbreviate $I_m$ as $I$.
Given a set of linear operators $S=\{M_1,M_2,...\}\subseteq \Xi(m)$, the algebra generated by $S$ is the set of linear operators which is a finite sum $\sum_k \alpha_k P_k$, where $\alpha_k\in \bC$ and $P_k$ is a product of finite operators in $S$.

 We use symbols $\cR,\cEN,\cF...$ to represent linear maps from $\Xi(m)$ to $\Xi(m)$. We use $\cI$ to denote the identity map.  We say a linear map $\cEN:\cH(m)\rightarrow \cH(m)$ is Completely Positive and Trace Preserving (CPTP) if there exists a set of linear operators $\{A_u\}_u$ such that $\forall M\in \cH(m), \cEN(M) = \sum_{u}A_u M A_u^\dagger $, where $\sum_{u}A_u^\dagger A_u =I$.
 
The trace norm of  $M\in \cH(m)$ is defined to be $\tn{M}:=tr(\sqrt{M^
\dagger M})$. The trace norm induces a norm on linear maps $\cR:\cH(m)\rightarrow \cH(m)$, which quantifies how much $\cR$ can scale the trace norm: \begin{align}
	\norm{\cR}:=\max_{M\in \cH(m); \tn{M}=1} \tn{\cR(M)}.
\end{align}

\section{Quantum Metropolis Algorithm in More Detail}
\label{sec:algo}

The algorithm takes in seven input parameters:
\begin{itemize}
\item An $n$-qubit local Hamiltonian $H$.
\item Inverse temperature $\beta$.
	\item $\tau$, which is a small $1/poly(n)$ real value that controls the weak measurement. 
	\item $K\in \bN$, which is the number of iterations in the main loop. 
        \item $\rho$, which is an arbitrary initial state.
	\item $r, g \in \bN$, which controls the precision of $\sQPE$. 
 \end{itemize}

  The jump operators  $\{C_j\}_j$, the acceptance operator $W$,  and the Boosted Quantum Phase Estimation ($\QPE,\FQPE$) are described in Section \ref{sec:algover}. For $\sQPE\in \{\QPE,\FQPE\}$, we use $\sQPE_{a,b}$ for applying $\sQPE$ on register $a,b$ of $n$ and $gr$ qubits respectively.

 The pseudo-code for the main algorithm is given in Algorithm \ref{alg:main}. All measurements are done in the computational basis. The outline of each iteration is given in Figure \ref{fig:iter}.

 \begin{algorithm}[!ht]
\caption{QMetropolis(H,$\beta$,$\tau$,$K$,$\rho$,$r$,$g$)}\label{alg:main}
\begin{algorithmic}[1]
\State  Initialize Register $1$ to $\rho$. \label{l1}
\Comment{For example one can set $\rho=\ket{0^n}\bra{0^n}$}
\For{$iter=1$ to $K$} 
	\State $\sQPE \leftarrow \{\QPE,\FQPE\}$ (with equal probability.) \label{sQPE}
	\State Sample $C\leftarrow \mu$,\label{lgate}  
	\State Append (fresh) Registers $2$, $3$, $4$ in state $\ket{0^{gr}}\ket{0^{gr}}\ket{0}$ ,\label{lInit}
 
	\State Define  $U := \sQPE_{1,3} \circ C \circ \sQPE_{1,2}$.\label{lQPE}
	\State Apply $U$, then apply $W$
        \State Measure register $4$ to get {\sc Outcome1}  \label{lU}
	\If {({\sc Outcome1} $ = 1$)} 
		\State {\bf Case} \acc: do nothing.
	\Else {~({\sc Outcome1} $= 0$)}  
		\State Apply($W^\dagger$), and measure register $4$ to get {\sc Outcome2}.
		\If {({\sc Outcome2} $= 1$)} 
			\State {\bf Case} \aacc: do nothing. 
		\Else {~({\sc Outcome2} $= 0$)} 
                \State {\bf Case} \rej: Apply $U^\dagger$
		\EndIf
	\EndIf
        \State Trace out (throw away) registers $2,3,4$.
\EndFor
\end{algorithmic}
\end{algorithm}

To illustrate the result of each step of the algorithm,  we give  explicit formulas for the contents of the registers throughout one iteration. Assume that the version of $\sQPE$ chosen is $\QPE$.
 Let $C$  be the chosen jump operator,  which can
 be expressed in the energy eigenbasis of $H$ as:
 \begin{align}
 	C\ket{\psi_j}=\sum_{k} 	c_{jk} \ket{\psi_k}.\label{eq:cjk}
 \end{align}
 Assume we begin with the state $\ket{\psi_j} \ket{0^{gr}}\ket{0^{gr}}\ket{0}$, where $\ket{\psi_j}$ is an eigenstate of $H$.
  Recall that:
 $$\QPE	\ket{\psi_j}\ket{0^{gr}} = \ket{\psi_j} \sum_{E\in S(r)^{\otimes g}} \beta_{jE} \ket{E}.$$
 Then:
 \begin{align}
  \QPE_{1,2} \ket{\psi_j}\ket{0^{gr}}\ket{0^{gr}} \ket{0} &= \sum_{E\in S(r)^{\otimes g}} \beta_{jE} \ket{\psi_j}\ket{E} \ket{0^{gr}} \ket{0},\\
 	C\cdot \QPE_{1,2} \ket{\psi_j}\ket{0^{gr}}\ket{0^{gr}} \ket{0} &= \sum_{k;E\in S(r)^{\otimes g}} \beta_{jE}\cdot c_{jk} \ket{\psi_k}\ket{E}\ket{0^{gr}} \ket{0},\\
 \QPE_{1,3} \cdot C\cdot \QPE_{1,2} \ket{\psi_j}\ket{0^{gr}}\ket{0^{gr}} \ket{0} &=  \sum_{k;E, E'\in S(r)^{\otimes g}} \beta_{jE}\cdot c_{jk} \cdot \beta_{kE'} \ket{\psi_k}\ket{E}\ket{E'} \ket{0}, \label{eq:algoform}
 \end{align} 
 Finally, when $W$ is applied, the result is
 \begin{equation}
     \sum_{k;E, E'\in S(r)^{\otimes g}} \beta_{jE}\cdot c_{jk} \cdot \beta_{kE'} \ket{\psi_k}\ket{E}\ket{E'} (\sqrt{1 - \tau^2 f_{EE'}} \ket{0} + \tau \sqrt{f_{EE'}}\ket{1})\label{eq:alg_state}
 \end{equation}
 Note that if $\FQPE$ is chosen instead of $\QPE$, the result is:
 \begin{equation}
 \sum_{k;E, E'\in S(r)^{\otimes g}} \beta^*_{jE}\cdot c_{jk} \cdot \beta^*_{kE'} \ket{\psi_k}\ket{E}\ket{E'} (\sqrt{1 - \tau^2 f_{EE'}} \ket{0} + \tau \sqrt{f_{EE'}}\ket{1})
 \end{equation}

Then for the state in Eq.~(\ref{eq:alg_state}) we measure register $4$ in computational basis.
 If we get measurement outcome $1$, the (unnormalized) state becomes
	\begin{align}   
 \sum_{k;E,E'\in S(r)^{\otimes g}} \tau \sqrt{f_{EE'}} \beta_{jE}\cdot c_{jk} \cdot \beta_{kE'}\ket{\psi_k} \ket{E}\ket{E'}\ket{1}. \label{eq:accept}
	\end{align}

\section{Main Theorems and the Effective Quantum Markov Chain}\label{sec:CTMC}

In this section, we state formal versions of the main theorems that we will prove about the performance of Algorithm \ref{alg:main}.
Note that each iteration in Algorithm \ref{alg:main} corresponds to a quantum channel which maps $n$-qubit states to $n$-qubit states. We denote this quantum channel as  $\cET$. We will expand $\cET$ as power series of $\tau$.
The performance of  Algorithm \ref{alg:main} can be analyzed by studying  the term  of order $\sim\tau^2$, which is $\cLt$ defined below. To ease notation, define
\begin{align}
     r_{\beta H}:= 1 +  \log \kappa_H + \log \beta.
 \end{align}
\begin{theorem}\label{thm:main}
	For any $n$-qubit state $\rho$, 
	\begin{align}
		\cET(\rho) =\left( \cI + \tau^2 \cLt  + \tau^4 \cJt[\tau]  \right)(\rho).\label{eq:evolution}
	\end{align}
	where $\cLt,\cJt[\tau]:\cH(n)\rightarrow \cH(n) $ are  linear maps on operators. $\cLt$ is independent of $\tau$. 
Assuming $r\geq r_{\beta H}$ we have
	\begin{itemize}
		\item \textbf{(Fixed point)} $\tn{\cLt(\rho_\beta)}\leq   \bL.$  
		\item \textbf{(Error terms)}   $\norm{\cJt[\tau]}\leq 4$. 
\item \textbf{(Uniqueness and Relaxation)}
	  There is a unique $\rho_\cLt$ such that $\cLt(\rho_\cLt)=0$. Besides, $\rho_
   \cLt$ is a full-rank quantum state and for any quantum state $\rho$, 
   \begin{align}
       \lim_{t\rightarrow \infty} e^{t\cLt}(\rho) = \rho_\cLt.
   \end{align}
	\end{itemize}
\end{theorem}
To make $\tn{\cLt(\rho_\beta)}\leq \delta$ parameters $g$ and $r$ are chosen so that $g=O(n+\log 1/\delta)$ and $r= O(\log \beta +\log \kappa_H+\log 1/\delta)$. The constant in the bound of $\tn{\cLt(\rho_\beta)}$  might be improved by a finer analysis. To ease notation, we will abbreviate $\cEN[\tau]$ as $\cEN$.

Theorem \ref{thm:main} suggests that Algorithm \ref{alg:main} 
effectively approximates a continuous-time chain $e^{t\cLt}$ with a  step-size of $\tau^2$. 
The divergence between the output of our algorithm and $\rho_\beta$ depends on the mixing time  of $\cLt$, which depends on the choice  of the jump operators $\{C_j\}_j$.

\begin{definition}[Mixing time]
Let $\epsilon$ be a precision parameter. The mixing time w.r.t $(\cLt,\epsilon)$ is defined to be the time needed for  driving any initial state $\epsilon$-close to its fixed point
\begin{align}
    t_{mix}(\cLt,\epsilon) := \inf \{ t\geq 0:   \tn{e^{t\cLt}(\rho)-\rho_\cLt}\leq \epsilon, \text{ for any quantum state } \rho\}.
\end{align}
\end{definition}

\begin{theorem}[Error bounds w.r.t Mixing time]\label{thm:mix}  
Let $\tau,\epsilon$ be parameters. Assume $r\geq r_{\beta H}$. For integer\footnote{For simplicity, here we assume $K:= t_{mix}(\cLt,\epsilon)/\tau^2$ is an integer. Otherwise we set $K$ to be the least integer which is greater than $t_{mix}(\cLt,\epsilon)/\tau^2$ and the error bounds can be analyzed similarly.} $K:= t_{mix}(\cLt,\epsilon)/\tau^2$, we have
   for any  quantum state $\rho$,
	$$\tn{\cEN^K(\rho) -\rho_\beta} \leq 2\epsilon +  \left( \bL +  2e^4 \tau^2\right)t_{mix}(\cLt,\epsilon).$$
\end{theorem}
Abbreviate $t_{mix}(\cLt,\epsilon)$ as $t_{mix}$. For  parameters 
\begin{align*}
   &\tau^2 = O(\epsilon/t_{mix})\\
   &g = O(n+\log t_{mix}+\log 1/\epsilon)\\
   &r = O(\log \beta +\log \kappa_H + \log t_{mix}+\log 1/\epsilon).\\
   &K= t_{mix}/\tau^2 = O(t_{mix}^2/\epsilon)
\end{align*}  
we have  $\tn{\cEN^K(\rho) -\rho_\beta}\leq 3\epsilon$.
The total runtime of the algorithm is 
$$K\cdot 4g\cdot poly(2^r,n)=O\left(poly(t_{mix}, \beta ,\kappa_H,1/\epsilon,n)\right),$$ 
which is $poly(n,1/\epsilon)$ assuming $t_{mix}=poly(n)$, $\beta\leq poly(n)$ and $\kappa_H\leq poly(n)$.  The $poly(2^r,n)$ in the above formula is mainly the cost of implementing one run of standard Quantum Phase estimation.  Theorem \ref{thm:main} and 
Theorem \ref{thm:mix}  will be proved in Section \ref{sec:eb_final}.


Instead of mixing time, as in the classical Metropolis algorithm,
one can also bound the error  $\tn{\cEN^K(\rho) -\rho_\beta} $ in terms of the spectral gap of $\cLt$. More precisely,
since $\cLt$ may not be Hermitian, we need to define a symmetrized version of $\cLt$ in order for the spectral gap to be well-defined.
Note that   $\sigma :=\rho_\cLt^{-1}$ is well-defined since $\rho_\cLt$ is full rank. Define  $\cLt^*$ to be dual map w.r.t inner product $$\langle M,N\rangle_\sigma:= tr(\sigma^{\frac{1}{2}}M^\dagger\sigma^{\frac{1}{2}}N).$$ Define the symmetrized map $\cLs=\frac{1}{2}(\cLt+\cLt^*)$. Then by definition  $\cLs$ is Hermitian w.r.t. $\langle ,\rangle_\sigma$ and is diagonalizable, thus its spectral gap, denoted as $\Upsilon$, is well-defined. Furthermore  one can prove $\cLs$ has a unique fixed point thus the spectral gap $\Upsilon$ is strictly greater than $0$. The following Theorem \ref{thm:mix_Ls_gap} implies $t_{mix}$ can be bounded in terms of $\Upsilon$. Thus one can translate Theorem \ref{thm:mix} in terms of $\Upsilon$ to get Corollary \ref{cor:eb_spe_gap}.
For completeness, we put a more detailed explanation of the dual map $\cLt^*$, the symmetrized map $\cLs$, their properties and a proof for  Theorem \ref{thm:mix_Ls_gap} in Appendix \ref{appendix:mixing}.

\begin{theorem}[Bounding mixing time w.r.t spectral gap]\label{thm:mix_Ls_gap} Define  $\sigma :=\rho_\cLt^{-1}$.
For any quantum state $\rho$, we have 
$$\tn{e^{t\cLt}(\rho) -\rho_\cLt}	 \leq 2^{n/2}  \cdot \sqrt{tr(\sigma^{\frac{1}{2}}\rho\sigma^{\frac{1}{2}}\rho)}\cdot \exp(-\Upsilon t ).$$
\end{theorem}

\begin{corollary}
Define  $\sigma :=\rho_\cLt^{-1}$, we have
	$$t_{mix}(\cLt,\epsilon)\leq \frac{1}{\Upsilon} \left( \ln \frac{1}{\epsilon} + \frac{n\ln 2}{2} + \frac{1}{2}\ln tr(\sigma^{\frac{1}{2}}\rho\sigma^{\frac{1}{2}}\rho) \right) $$.
\end{corollary}

\begin{corollary}
	[Error bounds w.r.t Spectral gap]\label{cor:eb_spe_gap}
Let $\tau,\epsilon$ be two parameters. Let $\Upsilon$ be the spectral gap of $\cLs$.
Assume $r\geq r_{\beta H}$. Define  $\sigma :=\rho_\cLt^{-1}$.
 Then for any  quantum state $\rho$, we have
 \begin{align}
     &\tn{\cEN^K(\rho) -\rho_\beta} \leq 2\epsilon + \left( \bL +  2e^4 \tau^2\right)\frac{1}{\Upsilon} \left( \ln \frac{1}{\epsilon} + \frac{n\ln 2}{2} + \frac{1}{2} \ln tr(\sigma^{\frac{1}{2}}\rho\sigma^{\frac{1}{2}}\rho) \right).\\
     &\text{for } K:= \frac{1}{\tau^2}  \frac{1}{\Upsilon} \left( \ln \frac{1}{\epsilon} + \frac{n\ln 2}{2} + \frac{1}{2}\ln tr(\sigma^{\frac{1}{2}}\rho\sigma^{\frac{1}{2}}\rho) \right)
 \end{align}
\end{corollary}

\textbf{Outline of this section.}
 In Section \ref{sec:CTMC_notations} we will define the operators corresponding to the three cases of the algorithm.
 In Section \ref{sec:operator_acc} to Section \ref{sec:operator_LC} we will derive the evolution equation
 Eq.~(\ref{eq:evolution}). In particular, we separate $\cEN$ into a sum of terms according to their dependence on $\tau$. This defines the operator $\cLt$ which we will analyze in later sections. Subsection \ref{sec:Lindblad} then writes $\cLt_\cC$ in the Lindbladian form. In Section \ref{sec:uniqueness} we prove the fixed point of $\cLt$ is unique.

\subsection{Definition of Operators for the Three Cases}\label{sec:CTMC_notations}

We will use the subscript $\cC$ to denote a particular choice for $\{C,\sQPE\}$. So, for example,
\begin{align}
    U_{\cC} = \sQPE_{1,3} \circ C \circ \sQPE_{1,2}.
\end{align}

We use  $(\Delta_0,\Delta_1)$ to denote measurement on  register $4$ in the computational basis, where
\begin{align}
	\Delta_0 &:= I_{n}\otimes I_{gr} \otimes I_{gr} \otimes \ket{0}\bra{0}.\\
	\Delta_1 &:= I_{n}\otimes I_{gr} \otimes I_{gr} \otimes \ket{1}\bra{1}.
\end{align}

If the state at the beginning of an iteration is $\rho$, then
the operations performed on $\rho \otimes \ket{0^{2gr+1}}\bra{0^{2gr+1}}$
in each of the three cases (before the last three registers are traced out)  can be summarized as:

\begin{align}
 \mbox{\acc:} ~~~  O_{a,\cC}	&: =   \Delta_1 W \Delta_0 ~\circ~ U_{\cC} \label{eq:Oa}\\
  \mbox{\aacc:} ~~~   O_{b,\cC} &:= \Delta_1 W^{\dag} \Delta_0 ~\circ~ \Delta_0 W \Delta_0 ~\circ~ U_{\cC} \label{eq:O1}\\
  \mbox{\rej:} ~~~   O_{r,\cC} &:= U^{\dag}_{\cC} ~\circ~ \Delta_0 W^{\dag} \Delta_0 ~\circ~ \Delta_0 W \Delta_0 ~\circ~ U_{\cC} \label{eq:orc}
\end{align}
The initial $\Delta_0$ is added in for symmetry and has no effect since the last register is always initialized as $\ket{0}$  at the start of each iteration. 
For $s \in \{a, b, r\}$ which represents \acc, \aacc\, and \rej,
the corresponding operator  which includes tracing out the last three registers is:

\begin{align}
  \cEN^{(s)}_{\cC}(\rho) & = tr_{2,3,4} \left(O_{s,\cC} ~~\left[ \rho \otimes \ket{0^{2gr+1}}\bra{0^{2gr+1}} \right] ~~ O^{\dag}_{s,\cC} \right)
\end{align}
 Where $\cEN = \cEN^{(a)} + \cEN^{(b)} + \cEN^{(r)}$ is the operator representing one iteration.
 In each of the next three subsections, we will derive alternative expressions for the operators in the three cases, as a sum of terms with different dependencies on the parameter $\tau$.
 Recall the definition of  $W$:
	\begin{align}
		W:= \sum_{E,E'\in S(r)^{\otimes g}} \ket{EE'}\bra{EE'}\otimes 	\begin{bmatrix}
 			\sqrt{1-\tau^2  f_{EE'}}& \tau  \sqrt{f_{EE'}}\\
 					\tau \sqrt{f_{EE'}} & -\sqrt{1-\tau^2 f_{EE'}}
 					\end{bmatrix}.
	\end{align}
Note that $W$ is Hermitian.
 \subsection{Operator for the \acc~ Case}\label{sec:operator_acc}

 \begin{definition}{\bf [Operators for the \acc~ Case]}
 \label{def:accops}
 \begin{align}
 W^{(10)} & := \sum_{E, E' \in S(r)^{\otimes g}} \sqrt{f_{EE'}} \ket{EE'} \bra{EE'} \otimes \ket{1}\bra{0}\\
\cM_{\cC}^{(a)}(\rho) & := tr_{2,3,4} \left(W^{(10)} U_{\cC} ~~\left[ \rho \otimes \ket{0^{2gr+1}}\bra{0^{2gr+1}} \right] ~~ U^{\dag}_{\cC} (W^{(10)})^{\dag}\right).
 \end{align}
 \end{definition}

 \begin{lemma}
     \label{lem:accops}
     $$\cEN^{(a)}_{\cC}(\rho) = \tau^2 \cM_{\cC}^{(a)}\text{ where } \norm{\cM_{\cC}^{(a)}}\leq 1.$$
 \end{lemma}

 \begin{proof}
     The main observation in proving the lemma is that $\Delta_1 W \Delta_0 =  \tau W^{(10)}$. Therefore
     \begin{align}
 \cEN^{(a)}_{\cC}(\rho) & = tr_{2,3,4} \left(O_{a,\cC} ~~\left[ \rho \otimes \ket{0^{2gr+1}}\bra{0^{2gr+1}} \right] ~~ O^{\dag}_{a,\cC} \right)\\
 & = tr_{2,3,4} \left(\Delta_1 W \Delta_0 ~\circ~ U_{\cC} ~~\left[ \rho \otimes \ket{0^{2gr+1}}\bra{0^{2gr+1}} \right] ~~ U^{\dag}_{\cC} ~\circ~ \Delta_0 W^{\dag} \Delta_1 \right)\\
 & = \tau^2 \cdot tr_{2,3,4}  \left(W^{(10)} ~\circ~ U_{\cC} ~~\left[ \rho \otimes \ket{0^{2gr+1}}\bra{0^{2gr+1}} \right] ~~ U^{\dag}_{\cC} ~\circ~ (W^{(10)})^{\dag} \right)\\
 &= \tau^2\cM_{\cC}^{(a)}(\rho).
\end{align}
Using Lemma \ref{lem:normJ_new} in Appendix \ref{appendix:Matrix_Norm}, we can observe that $\norm{\cM_{\cC}^{(a)}}\leq 1$.
 \end{proof}

 \subsection{Operator for the \aacc~ Case}\label{sec:operator_accc}

 \begin{lemma}
     \label{lem:aaccops}
     $$\cEN^{(b)}_{\cC}(\rho) = \tau^2 \cM_{\cC}^{(a)} + \tau^4 \cJt^{(b)}_{\cC}[\tau],$$
     where $\norm{\cJt^{(b)}_{\cC}[\tau]} \le 3$ for $\tau\in [0,1]$.
 \end{lemma}

 \begin{proof}
 Define the function $W^{(00)}$ to be
     $$W^{(00)}  = \sum_{E, E' \in S(r)^{\otimes g}} \frac{f_{EE'}}{\sqrt{1 - \tau^2 f_{EE'}}+ 1} \ket{EE'} \bra{EE'} \otimes \ket{0}\bra{0},$$
     Note that $f_{EE'}/(\sqrt{1 - \tau^2 f_{EE'}}  + 1) \in [0,1]$ since both $\tau$ and $f_{EE'}$ are in $[0,1]$.
     With this definition in place, observe that
      $$\Delta_0 W \Delta_0 = \Delta_0 - \tau^2 W^{(00)},$$ and
     $$\Delta_1 W^{\dag} \Delta_0 ~\circ~ \Delta_0 W \Delta_0 =   \tau W^{(10)} ( \Delta_0 - \tau^2 W^{(00)})
     =  \tau W^{(10)} -  \tau^3 W^{(10)} W^{(00)}$$
     Now we can express $\cEN^{(b)}_{\cC}(\rho)$ as
     $$tr_{2,3,4} \left(  \left( \tau W^{(10)} -  \tau^3 W^{(10)} W^{(00)} \right) U_{\cC} ~~\left[ \rho \otimes \ket{0^{2gr+1}}\bra{0^{2gr+1}} \right] ~~ U^{\dag}_{\cC} \left(  \tau W^{(10)} - \tau^3 W^{(10)} W^{(00)} \right)^{\dag} \right)$$
     In multiplying out the terms, there is one $\tau^2$ term:
     $$\tau^2 tr_{2,3,4} \left(    W^{(10)}  U_{\cC} ~~\left[ \rho \otimes \ket{0^{2gr+1}}\bra{0^{2gr+1}} \right] ~~ U^{\dag}_{\cC}   (W^{(10)} )^{\dag} \right),$$
     which is equal to $\tau^2 \cM_{\cC}^{(a)}$.
     There are two $\tau^4$ terms:
$$- \tau^4 tr_{2,3,4} \left(     W^{(10)} W^{(00)}  U_{\cC} ~~\left[ \rho \otimes \ket{0^{2gr+1}}\bra{0^{2gr+1}} \right] ~~ U^{\dag}_{\cC}   (W^{(10)})^{\dag} \right)$$
$$- \tau^4 tr_{2,3,4} \left(     W^{(10)}  U_{\cC} ~~\left[ \rho \otimes \ket{0^{2gr+1}}\bra{0^{2gr+1}} \right] ~~ U^{\dag}_{\cC}   (W^{(10)} W^{(00)})^{\dag} \right)$$
     Finally, there is one $\tau^6$ term:
     $$ \tau^6 tr_{2,3,4} \left(     W^{(10)} W^{(00)}  U_{\cC} ~~\left[ \rho \otimes \ket{0^{2gr+1}}\bra{0^{2gr+1}} \right] ~~ U^{\dag}_{\cC}   (W^{(10)} W^{(00)})^{\dag} \right)$$
     The sum of the two $\tau^4$ term and the $\tau^6$ term is denoted as $\tau^4\cdot \cJt^{(b)}_{\cC}[\tau]$.
     Using Lemma \ref{lem:normJ_new} in Appendix \ref{appendix:Matrix_Norm}, we can observe that $\norm{\cJ^{(b)}_{\cC}[\tau]} \le 3$ for $\tau\in[0,1]$. 
 \end{proof}

 \subsection{Operator for the \rej~ Case}\label{sec:operator_rej}

 \begin{definition}{\bf [Operators for the \rej~ Case]}
 \label{def:rejops}
 \begin{align}
 W^{(000)} & := \sum_{E, E' \in S(r)^{\otimes g}} f_{EE'} \ket{EE'} \bra{EE'} \otimes \ket{0}\bra{0} = (W^{(10)})^\dagger W^{(10)} \\
\cM_{\cC}^{(r)}(\rho) & :=  \bra{0^{2gr+1}} U^{\dag}_{\cC} W^{(000)} U_{\cC} \ket{0^{2gr+1}} \cdot \rho + \rho \cdot
\bra{0^{2gr+1}} U^{\dag}_{\cC} W^{(000)} U_{\cC} \ket{0^{2gr+1}} 
 \end{align}
 \end{definition}

 \begin{lemma}
     \label{lem:rejops}
     $$\cEN^{(r)}_{\cC}(\rho) = \cI - \tau^2 \cM_{\cC}^{(r)} + \tau^4 \cJt^{(r)}_{\cC},$$
     where $\norm{\cM_{\cC}^{(r)}}\leq 2$ and  $\norm{\cJt^{(r)}_{\cC}} \le 1$.
 \end{lemma}

 \begin{proof}
     First observe that
      $$\Delta_0 W \Delta_0 ~\circ~ \Delta_0 W \Delta_0= \Delta_0 - \tau^2 W^{(000)}.$$
      Also $W^{(000)}$ is Hermetian.  Now we can express $\cEN^{(r)}_{\cC}(\rho)$ as
     $$tr_{2,3,4} \left(  U^{\dag}_{\cC} \left( \Delta_0 - \tau^2 W^{(000)} \right) U_{\cC} ~~\left[ \rho \otimes \ket{0^{2gr+1}}\bra{0^{2gr+1}} \right] ~~ U^{\dag}_{\cC} \left( \Delta_0 - \tau^2 W^{(000)} \right)^{\dag} U_{\cC} \right)$$
     There is one term independent of $\tau$ which is the identity since $\Delta_0$ and $U_{\cC}$ commute:
     $$tr_{2,3,4} \left(  U^{\dag}_{\cC} \Delta_0  U_{\cC} ~~\left[ \rho \otimes \ket{0^{2gr+1}}\bra{0^{2gr+1}} \right] ~~ U^{\dag}_{\cC} \Delta_0  U_{\cC} \right) = \cI(\rho)=\rho$$
     There are two $\tau^2$ terms:
     $$- \tau^2 tr_{2,3,4} \left(  U^{\dag}_{\cC}  W^{(000)}  U_{\cC} ~~\left[ \rho \otimes \ket{0^{2gr+1}}\bra{0^{2gr+1}} \right]  \right) = - \tau^2 \bra{0^{2gr+1}} U^{\dag}_{\cC} W^{(000)} U_{\cC} \ket{0^{2gr+1}} \cdot \rho$$ 
     $$- \tau^2 tr_{2,3,4} \left(   ~~\left[ \rho \otimes \ket{0^{2gr+1}}\bra{0^{2gr+1}} \right] U^{\dag}_{\cC}  W^{(000)}  U_{\cC} \right) = - \tau^2 \rho \cdot \bra{0^{2gr+1}} U^{\dag}_{\cC} W^{(000)} U_{\cC} \ket{0^{2gr+1}} $$
     The sum of the $\tau^2$ terms is equal to $- \tau^2 \cM_{\cC}^{(r)}(\rho)$.
     Finally there is a $\tau^4$:
     $$\tau^4 tr_{2,3,4} \left(  U^{\dag}_{\cC}  W^{(000)}  U_{\cC} ~~\left[ \rho \otimes \ket{0^{2gr+1}}\bra{0^{2gr+1}} \right] ~~ U^{\dag}_{\cC}  W^{(000)}  U_{\cC} \right)$$
     This last term is defined to be $\cJt^{(r)}_{\cC}$. Using Lemma \ref{lem:normJ_new} in Appendix \ref{appendix:Matrix_Norm}, we can observe that $\norm{\cM_{\cC}^{(r)}}\leq 2$ and $\norm{\cJt^{(r)}_{\cC}} \le 1$ for $\tau\in [0,1]$.
 \end{proof}

 \subsection{The Definition of Operator $\cLt_{\cC}$ and $\cLt$}\label{sec:operator_LC}

  \begin{definition}{\bf [The Operators $\cLt_{\cC},\cJ_{\cC} $ and $\cLt,\cJ$]}
 \label{def:L}
 \begin{align}
 \cLt_{\cC}  &:= 2 \cM_{\cC}^{(a)} - \cM_{\cC}^{(r)}\\
 \cJt_{\cC}[\tau]  &:=  \cJt_{\cC}^{(b)}[\tau] + \cJt_{\cC}^{(r)}
  \end{align}
  Averaging the random choices for $\sQPE$ and the jump operator $C$, 
  \begin{align}
 \cLt  &:= \sum_{\cC=\{C,\sQPE\}} \frac{1}{2} \mu(C) \cLt_\cC.
 \end{align}
 where the $\frac{1}{2}$ comes from the fact that $\sQPE$ is chosen uniformly from $\{\QPE,\FQPE\}$. Define $\cM^{(a)}, \cM^{(r)}$ and $\cJt[\tau]$ from 
$\cM^{(a)}_\cC, \cM^{(r)}_\cC$ and $\cJt_\cC[\tau]$
 similarly.
 \end{definition}

 \begin{lemma}\label{lem:evol_main}
 \begin{align}
        &\cEN = \cI + \tau^2 \cLt + \tau^4 \cJt[\tau],
 \end{align}
     where $\norm{\cLt}\leq 4$, and $\norm{\cJt[\tau]} \le 4$ for $\tau\in [0,1]$.
 \end{lemma}

 \begin{proof}
 By Lemmas \ref{lem:accops}, \ref{lem:aaccops}, and \ref{lem:rejops},
     \begin{align}
         \cEN_{\cC} & = \cEN_{\cC}^{(a)} + \cEN_{\cC}^{(b)} + \cEN_{\cC}^{(r)}\\
         & = \tau^2 \cM_{\cC}^{(a)} + (\tau^2\cM_{\cC}^{(a)} + \tau^4\cJt_{\cC}^{(b)}[\tau]) + (\cI - \tau^2\cM_{\cC}^{(r)} + \tau^4 \cJt_{\cC}^{(r)})\\
         & = \cI + \tau^2 (2 \cM_{\cC}^{(a)} - \cM_{\cC}^{(r)}) + \tau^4 (\cJt_{\cC}^{(b)}[\tau] + \cJt_{\cC}^{(r)})\\
         & = \cI + \tau^2 \cLt_{\cC} + \tau^4 \cJt_{\cC}[\tau].     \end{align}
     Since $\cJt_{\cC}  =  \cJt_{\cC}^{(b)}[\tau] + \cJt_{\cC}^{(r)}$, $\norm{\cJt^{(b)}_{\cC}[\tau]} \le 3$, and $\norm{\cJ^{(r)}_{\cC}} \le 1$, the bound on $\norm{\cJ_{\cC}[\tau]}$ follows by triangle inequality. Similarly $\norm{\cLt_\cC} \leq 2 \norm{\cM^{(a)}_\cC} + \norm{\cM^{(r)}_\cC} \leq 4$.
     The equation for  $\cEN$  comes from linearity.
     The bounds for $\norm{\cJ[\tau]}$ and $\norm{\cLt}$ come from triangle inequality.
 \end{proof}

\section{Uniqueness of the Fixed Point}\label{sec:Lind_unique}
In this section, we rewrite $\cLt$ in the Lindbladian form, and prove that $\cLt$ has a unique fixed point if the jump operators $\{C_j\}_j$ generate the full algebra.   This  section is independent of Section \ref{sec:fp} which proves $\cLt$ approximately fixes the Gibbs state, and can be skipped temporarily. 

\subsection{Lindbladian form of $\cLt$}\label{sec:Lindblad}

\begin{lemma}\label{lem:Lind_form}
    $\cLt_\cC$  can be written in Lindbladian form, that is defining
$$ S_\cC(EE'z) :=   \bra{EE'z} W^{(10)} ~\circ~ U_{\cC} \ket{0^{2gr+1}}$$
we have
   \begin{align}
    \cLt_\cC(\rho) = \sum_{EE'z} 2\cdot S_\cC(EE'z) \cdot \rho \cdot  S_\cC(EE'z)^\dagger -\left\{ S_\cC(EE'z)^\dagger S_\cC(EE'z),  \rho  \right\}_+.
    \end{align}
\end{lemma}
\begin{proof}
One can check that
\begin{align}
	\cLt_\cC(\rho) &= 2\cM_\cC^{(a)}(\rho) - \cM_\cC^{(r)}(\rho)\\
				&=  \sum_{EE'z} 2\cdot S_\cC(EE'z)\cdot  \rho \cdot S_\cC(EE'z)^\dagger \\
    &-  \sum_{EE'z} \left(S_\cC(EE'z)^\dagger \cdot S_\cC(EE'z) \cdot \rho +    \rho \cdot S_\cC(EE'z)^\dagger S_\cC(EE'z)\right).\\
				&= \sum_{EE'z} 2\cdot S_\cC(EE'z) \cdot \rho \cdot  S_\cC(EE'z)^\dagger -\left\{ S_\cC(EE'z)^\dagger S_\cC(EE'z),  \rho  \right\}_+ \label{eq:merge}
\end{align}
\end{proof}

In the following we give a sketch that  $\cLt_\cC$ and the Davies generator $\cD_{\alpha}(w)$ have similar forms. This observation is just for intuition and will not be used in any proof. 
Due to the in-deterministic and imperfect energy estimation of $\QPE$, it is unclear whether  the proof techniques used for  Davies-generator-based Gibbs sampler~\cite{chen2023quantum,ding2024efficient} can be adapted to show that our $\cLt$ satisfying $\cLt(\rho_\beta)\approx0$. These proof techniques~\cite{chen2023quantum,ding2024efficient}  are based on  bounding the approximation error by truncating the infinite integral in the (weighted) Davies generator to a finite region.

 Recall that the canonical form of Davies generator $\cD$ w.r.t. jump operators $\{A_\alpha\}_\alpha$, in the Schrodinger picture, is given by
\begin{align}
\cD(\rho) &= -i [H,\rho] + \sum_{w,\alpha} \cD_{\alpha}(w) (\rho)\label{eq:sumD} \\ 	
\cD_{\alpha}(w) (\rho) &= G^\alpha(w)  \left( 2\cdot A_{\alpha}(w) \cdot \rho\cdot   A_{\alpha}(w)^\dagger - \left\{A_{\alpha}(w)^\dagger A_{\alpha}(w),\rho\right\}_+    \right)\\
A_\alpha(w) &:= \int_{-\infty}^{+\infty} e^{iwt} e^{-iHt} A_\alpha e^{iHt} dt.
\end{align}
Here $G^\alpha(w)$ is the acceptance rate, $A_\alpha$ is the jump operator, and $A_\alpha(w)$ is  an operator which maps a state of energy $\nu$ to energy $\nu+w$.  The summation $\sum_{w}$ sums over all possible energy difference $\{E_j-E_k\}_{j,k}$. 
The $S_\cC(EE'z)$ in our $\cLt$ is a conceptual analog of $\sqrt{G^{\alpha}} A_\alpha(w)$, which map states with energy approximately $\overline{E}$ to states with energy approximately $\overline{E'}$.
The jump operator $C$ is an analog of $A_\alpha$.

\subsection{Uniqueness of the Fixed Point}\label{sec:uniqueness}
In this subsection, we prove the following theorem.

\begin{theorem}\label{thm:uniqueness}
[Uniqueness of full-rank fixed point]
Suppose that in Algorithm \ref{alg:main},  the algebra generated by jump operators $\{C_j\}_j$  is equal to the full algebra, that is, the set of all $n$-qubit operators.
Then there is a  unique $\rho_\cLt\in\Xi (n)$ such that $\cLt(\rho_\cLt)=0$. In addition, $\rho_
   \cLt$ is a full-rank quantum state, and for any quantum state $\rho$, 
   \begin{align}
       \lim_{t\rightarrow \infty} e^{t\cLt}(\rho) = \rho_\cLt.
   \end{align}
\end{theorem}

We first prove a Lemma.
\begin{lemma}\label{lem:eLCPTP}
    $e^{t\cLt}$ is CPTP for any $t\geq 0$.
\end{lemma}
\begin{proof}
    Note that  
\begin{align}
	e^{t\cLt }=\lim_{\delta\rightarrow 0} \mathcal{E}[\delta]^{t/\delta^2}. 
\end{align}
where the limit is taking by decreasing $\delta\geq 0$ to $0$. 
Since $\mathcal{E}[\delta]$ is CPTP  for any $\delta\in[0,1]$,  $e^{t\cLt }$ is also CPTP. An alternative proof can be obtained by noticing that  Lemma \ref{lem:Lind_form} implies that $\cLt$ satisfies Theorem \ref{thm:CPTP} in Appendix \ref{appendix:mixing}.
\end{proof}

We invoke the following Theorems to prove the uniqueness of the fixed point, which can be adapted from Corollary 7.2 of \cite{wolf2012quantum} or  Lemma 2 in \cite{ding2024efficient}.

\begin{theorem}[\cite{wolf2012quantum,ding2024efficient}]\label{thm:wol12}
    Suppose $e^{\cP t}:\Xi(n)\rightarrow \Xi(n)$ is a CPTP map for any $t\geq 0$, with generator 
    \begin{align}
\cP(\rho)	= - i [H_{system},\rho] + \sum_{k\in S}\left( V_k \rho V_k^\dagger -\frac{1}{2}\left\{ V_k^\dagger V_k,\rho \right\}_+  \right). 
\end{align}
If the algebra generated by operators $\{V_k\}_k$ is the full algebra $\Xi(n)$. Then there exists a unique $\rho_\cP$ such that $\cP(\rho_\cP)=0$. In addition, $\rho_
   \cP$ is a full-rank quantum state, and for any quantum state $\rho$, 
   \begin{align}
       \lim_{t\rightarrow \infty} e^{t\cP}(\rho) = \rho_\cP.
   \end{align}
\end{theorem}

We use the above Theorem to prove that $\cLt$ has a unique fixed point.

\begin{proof}[Proof of Theorem \ref{thm:uniqueness}]  
By Lemma \ref{lem:eLCPTP} we know that  $e^{t\cLt }$ is CPTP. 
Then we verify the conditions in Theorem \ref{thm:wol12}.
Recall that in Lemma \ref{lem:Lind_form} we have written $\cLt_\cC$ in  Lindbladian form in terms of $S_\cC(EE'z)$. By definition
\begin{align}
    \cLt &=\sum_{\cC=\{C,\sQPE\}}\frac{1}{2}\mu(C)\cdot \cLt_\cC\label{eq:fix_1}.
\end{align}
Define
		\begin{align}
	 H_{system} &=0\\ 
		 V_\cC(EE'z)  &:= \sqrt{\mu(C)}\cdot S_\cC(EE'z)\label{eq:V}.
	\end{align}
One can check that 
\begin{align}
\cLt =
- i [H_{system},\rho] + \sum_{\cC, EE'z} \left( V_\cC(EE'z) \rho V_\cC(EE'z)^\dagger -\frac{1}{2}\left\{ V_\cC(EE'z)^\dagger V_\cC(EE'z),\rho \right\}_+  \right). 	
\end{align}

Recall that in Section \ref{sec:algo} Eq.~(\ref{eq:accept}) we have computed 
\begin{align}
	\ket{\eta_{j\{C,\QPE\}}}  &:= \tau\cdot W^{(10)} \circ U_{C,\QPE} \ket{\psi_j,0^{2gr+1}} \\
&=  \sum_{k;E,E'\in S(r)^{\otimes g}} \tau \sqrt{f_{EE'}} \beta_{jE}\cdot c_{jk} \cdot \beta_{kE'}\ket{\psi_k} \ket{E}\ket{E'}\ket{1}.
\end{align}
 One can check that 
	\begin{align}
	 V_{\{C,\QPE\}}(EE'z)  &=  \sqrt{\mu(C)}\cdot S_\cC(EE'z)\\
	  & =  \sqrt{\mu(C)}\cdot \bra{EE'z} W^{(10)}\circ U_{\{C,\QPE\}} \sum_j\ket{\psi_j}\bra{\psi_j}\otimes \ket{0^{2gr+1}}\\
	  &= \sqrt{\mu(C)}\cdot 
	   \sum_j  \frac{1}{\tau}\langle EE'z \ket{\eta_{j\{C,\QPE\}}} \bra{\psi_j}\\
		&=  \sqrt{\mu(C)f_{EE'}} \cdot \delta_{z1} \sum_{j,k}   \beta_{jE}\cdot c_{jk} \cdot \beta_{kE'}\ket{\psi_k}\bra{\psi_j}\label{eq:V_form}
	\end{align}

Define a  matrix $B_E\in \Xi(n)$ such that
\begin{align}
	B_E \ket{\psi_j} =\beta_{jE} \ket{\psi_j},
\end{align}
From Eq.~(\ref{eq:V_form})
one can check that
\begin{align}
	 V_{\{C,\QPE\}}(EE'1) = \sqrt{\mu(C)f_{EE'}}\cdot B_{E'}\cdot C\cdot B_{E}
\end{align}
From Eq.~(\ref{eq:bjE1}) in Appendix \ref{appendix:QPE} we know that
 \begin{align}
	&\sum_E B_E =I_n,\label{eq:IE}\\
 &\sum_{EE'} \frac{1}{\sqrt{\mu(C)f_{EE'}}} V_{\{C,\QPE\}}(EE'1) =
 \left(\sum_{E'} B_{E'} \right) \cdot C \cdot 	 \left(\sum_{E} B_{E} \right) =C.
 \end{align}
Since  the algebra generated by $\{C_j\}_j$ is the full algebra $\Xi(n)$, thus the algebra generated by  $\{V_\cC(EE'z)\}_{\cC,E,E',z}$ is the full algebra $\Xi(n)$. We can then use Theorem \ref{thm:wol12} we complete the proof.

\end{proof}

\section{Gibbs States as Approximate Fixed Point}\label{sec:fp}

In this section, we will prove $\cLt$ approximately fixes the Gibbs state, that is
\begin{lemma}\label{lem:fixed_point}
If $r\geq r_{\beta H}$, we have
	$$\tn{\cLt(\rho_{\beta})}\leq \bL.$$
\end{lemma}

The outline of the proof is as follows. In section \ref{sec:fp_trun_gibbs} we define a truncated Gibbs states $\rho_{\beta0}$ where  $\tn{\rho_\beta-\rho_{\beta0}}\approx 0$, thus $\tn{\cLt(\rho_{\beta})} \approx \tn{\cLt(\rho_{\beta0})}$ since $\norm{\cLt}$ is bounded. The remaining subsections then focus on bounding $\tn{\cLt(\rho_{\beta0})}$. Then in Section  \ref{sec:fp_trun_QPE} and   Section  \ref{sec:fp_trun_ope} we define projected operators which will be used as auxiliary notations in the following proofs. In \ref{sec:fp_trun_uni} we show that $\cLt(\rho_{\beta0})$ can be written as the sum of constant number of matrices, where each matrix has small trace norm. Finally in Section \ref{sec:fp_trun_fin} we complete the proof of Lemma \ref{lem:fixed_point}.

\subsection{Truncated Energy and Truncated Gibbs States}\label{sec:fp_trun_gibbs}

Recall that from Section \ref{sec:algover}, that Boosted Quantum Phase Estimation acts as:
$$
 \QPE	\ket{\psi_j}\ket{0^{gr}}= \ket{\psi_j} \sum_{E \in S(r)^{\otimes g}} \beta_{jE} \ket{E}
$$
where $S(r)$ is the set of  energy estimations that  are integer multiples of 
$$
w:= \kappa_H\cdot 2^{-r}.
$$
We have defined $r_{\beta H} = 1 +  \log \kappa_H + \log \beta$, so
for $r\geq r_{\beta H}$, we have $2\beta w\leq 1$.
For any real value $\nu\geq 0$, 
 $\flo{\nu}$  
denotes the closet value to $\nu$ which is an integer multiple of $w$ and is smaller or equal to $\nu$. For any $k\in \bN$, define
\begin{align}
&\nu^{(k)} := \flo{\nu} + k w.\\
\text{ in particular }&E_{j}^{(k)} := \flo{E_j} + k w.	
\end{align}
Recall that $Z=tr(\exp(-\beta H))$ is the partition function of $\rho_\beta$, and $p_j=\exp(-\beta E_j)/Z$ is the corresponding probability. Define the truncated probability and truncated Gibbs states as: 
\begin{align}
	p_{jk} &:= \exp{(-\beta  E_j^{(k)}})/Z,\quad \text{ for $k\in \bN$.}\\
	\rho_{\beta0}  &:= \sum_j p_{j0} \ket{\psi_j}\bra{\psi_j}.
 \end{align}
One can check that
\begin{lemma}\label{lem:pq_bound_QPEr}
If $2\beta w\leq 1$,
then $ |p_j-p_{j0}|\leq p_j \cdot 2\beta w$ and $   
	\tn{\rho_{\beta0}-\rho_\beta}\leq 2\beta w.$
\end{lemma}

Lemma \ref{lem:evol_main} proves that $\norm{\cLt}\leq 4$.
Using this fact, along with the 
 triangle inequality and Lemma \ref{lem:pq_bound_QPEr},  we have 
\begin{lemma}\label{lem:fp_trun}
   $$ \tn{\cLt(\rho_\beta)}\leq \tn{\cLt(\rho_{\beta0})} + \tn{\cLt(\rho_\beta-\rho_{\beta0})}\leq \tn{\cLt(\rho_{\beta0})} + 4\cdot 2\beta w.$$
\end{lemma}
The remainder of this section focuses on bounding $\tn{\cLt(\rho_{\beta0})}$.

\subsection{Projected $\QPE$.}
\label{sec:fp_trun_QPE}
From Lemma \ref{lem:QPEerr}, we know that the median estimation from 
$\QPE$  almost always maps to one of two possible values: $E_j^{(0)}=\flo{E_j}$ or $E_j^{(1)}=\cei{E_j}$.
A remark is that the number ``two" is not important in the analysis as long as it is a constant. 
We define the following  projections that separate out the cases depending on the output of $\QPE$: 
\begin{align}
    &P^{(k)}:= \sum_j \ket{\psi_j}\bra{\psi_j}\otimes \sum_{E\in S(r)^{\otimes g}: \,\overline{E}=E_j^{(k)}}  \ket{E}\bra{E}, \text{ for $k\in\{0,1\}$.}\\
     &P^{(else)}:= I - P^{(0)}-P^{(1)}.
\end{align}
To analyze the performance of our algorithm, we decompose $\QPE$ into three operators, according to the above projections. That is 
\begin{align}
& \nI:=\{0,1,``else"\}.\\
	&\QPE= \sum_{k\in \nI} \QPE^{(k)}\\
 \text{where }&\QPE^{(k)}:= P^{(k)}\cdot \QPE.\\
 &\QPE^{(k)}	\ket{\psi_j}\ket{0^{gr}} = \ket{\psi_j} \sum_{E} \beta_{jE}^{(k)} \ket{E},\\
&\beta_{jE}^{(k)} :=\left\{ 
 \begin{aligned}
     \beta_{jE},  &\quad\text{ if $\overline{E}= E_{j}^{(k)}$},\\
     0, &\quad\text{ else.}
 \end{aligned}
 \right.\quad \quad\text{ for $k\in\{0,1\}$ }\\
 &\beta_{jE}^{else} :=\left\{ 
 \begin{aligned}
     \beta_{jE},  &\quad\text{ if $\overline{E}\not\in \{ E_j^{(0)},E_j^{(1)}\}.$}\\
     0, &\quad\text{ else.}
 \end{aligned}
 \right.
\end{align}

For convenience, for any subset $A\subseteq T$, we also define
\begin{align}
    P^{(A)}:=\sum_{k\in A} P^{(k)}.
\end{align}
and define $\QPE^{(A)}$, $\beta_{jE}^{(A)}$ accordingly. 
For $\FQPE$ we similarly define $\FQPE^{(k)}$ and $\FQPE^{(A)}$.

\subsection{Projected Operators}\label{sec:fp_trun_ope}

Recall that  the $\cLt$ for our algorithm is defined as
\begin{align*}
&\cLt_\cC = 2\cM_\cC^{(a)} -\cM_\cC^{(r)}.\\
&\cM_\cC^{(a)}(\rho) =  tr_{2,3,4} \left( W^{(10)}U_\cC ~[\rho\otimes \ket{0^{2gr+1}}\bra{0^{2gr+1}}]~ U_\cC^\dagger (W^{(10)})^\dagger \right)	\\
	&\cM_\cC^{(r)}(\rho) =   \bra{0^{2gr+1}} U_\cC^\dagger (W^{(10)})^\dagger W^{(10)} U_\cC\ket{0^{2gr+1}}\rho + \rho\bra{0^{2gr+1}}  U_\cC^\dagger (W^{(10)})^\dagger W^{(10)} U_\cC\ket{0^{2gr+1}}.
\end{align*}
where
\begin{align}	
	&W^{(10)}  = \sum_{E, E' \in S(r)^{\otimes g}} \sqrt{f_{EE'}} \ket{EE'} \bra{EE'} \otimes \ket{1}\bra{0}.\\
	& U_\cC = \sQPE_{1,3} \circ C \circ \sQPE_{1,2}.
\end{align}
For convenience, we divide $\cM_\cC^{(r)}$ into ``right" and ``left" terms 
\begin{align}
	&\cM_\cC^{(rr)}(\rho) =   \bra{0^{2gr+1}}  U_\cC^\dagger (W^{(10)})^\dagger W^{(10)} U_\cC\ket{0^{2gr+1}}\rho\\
	&\cM_\cC^{(rl)}(\rho) =  \rho\bra{0^{2gr+1}}   U_\cC^\dagger (W^{(10)})^\dagger W^{(10)} U_\cC\ket{0^{2gr+1}}
\end{align}
By the following Lemma it suffices to only analyze the right term $\cM^{(rr)}$.
\begin{lemma}\label{lem:rl2rr}
For any Hermitian $\rho$,
    \begin{align}
    \tn{\cLt(\rho)}\leq 2  \tn{ \cM^{(a)}(\rho)-   \cM^{(rr)}(\rho)}.
    \end{align}
\end{lemma}
\begin{proof}
Note that 
 $\cM^{(a)}(\rho)=\cM^{(a)}(\rho)^\dagger$,  $\cM^{(rl)}(\rho)=\cM^{(rr)}(\rho)^\dagger$. Since  $\tn{N^\dagger}=\tn{N}$ for any matrix $N$, we have 
    $$ \tn{ \cM^{(a)}(\rho)-   \cM^{(rl)}(\rho)} = \tn{ \cM^{(a)}(\rho)-   \cM^{(rr)}(\rho)}.$$ Thus 
    $$\tn{\cLt(\rho)}\leq \tn{ \cM^{(a)}(\rho)-   \cM^{(rl)}(\rho)} + \tn{ \cM^{(a)}(\rho)-   \cM^{(rr)}(\rho)} \leq 2\tn{ \cM^{(a)}(\rho)-   \cM^{(rr)}(\rho)}. $$
\end{proof}

\noindent
We can now use the decomposition of $\QPE$ to decompose the accept and reject operators.

\begin{definition}[Projected Operators]
For any subsets $A,B\subseteq \nI$, we define $U_{\cC}^{(AB)}$  by substituting $\sQPE$ with corresponding operators:
\begin{align}
&U_\cC^{(AB)} := \sQPE^{(A)}_{1,3} \circ C \circ \sQPE^{(B)}_{1,2}.
\end{align}
Accordingly for subsets $A,B,X,Y\subseteq \nI$, we define
\begin{align}
 &\cM_\cC^{(a,ABXY)}(\rho) :=  tr_{2,3,4} \left( W^{(10)}U^{(AB)}_\cC ~[\rho\otimes \ket{0^{2gr+1}}\bra{0^{2gr+1}}]~ (U_\cC^{(XY)})^\dagger (W^{(10)})^\dagger \right)		\\
	&\cM_\cC^{(rr,ABXY)}(\rho) :=   \bra{0^{2gr+1}} (U_\cC^{(BA)})^\dagger (W^{(10)})^\dagger W^{(10)}U_\cC^{(YX)} \ket{0^{2gr+1}}\rho.\label{eq:uvst1}
	\end{align}
 \end{definition}

 \vspace{.1in}
 \noindent
Note that in Eq.~(\ref{eq:uvst1}) we use $U_\cC^{(BA)},U_\cC^{(YX)}$ instead of $U_\cC^{(AB)},U_\cC^{(XY)}$. 

 \vspace{.1in}
 
    One can directly check that the following  Lemma are true, which we omit the proofs.
\begin{lemma}\label{lem:eq_uvst}
	\begin{align}
  \cM_\cC^{(a)} = \sum_{u,v,s,t\in \nI}  \cM_\cC^{(a,uvst)},\quad
  \cM_\cC^{(rr)} = \sum_{u,v,s,t\in \nI}  \cM_\cC^{(rr,uvst)}.
	\end{align}
 Similarly after averaging over the random selection of $\cC$, we get
	\begin{align}
 &  \cM^{(a)} = \sum_{u,v,s,t\in \nI}  \cM^{(a,uvst)},  \quad \cM^{(rr)} = \sum_{u,v,s,t\in \nI}  \cM^{(rr,uvst)}.
	\end{align}
\end{lemma}

\subsection{Explicit Expressions for Operators in the Accept and Reject Cases}

The following two lemmas give explicit representations for the operators $\cM^{(a,ABXY)}$ and $\cM^{(rr,ABXY)}$
in terms of the energy eigenstates $\{\ket{\psi_j}\}_j$.
To ease notations, we write $(\beta_{jE}^{(\cdot)})^*$ as $\beta_{jE}^{(\cdot)*}$.

\begin{lemma} For any subsets $A,B,X,Y\subseteq \nI$,
\label{lem:expr_rp_QPE}
         \begin{align}
           &\bra{ \psi_m} \cM^{(a,ABXY)}(\rho_{\beta 0})\ket{\psi_k} 
           =  \sum_j p_{j0} \sum_{E, E'} f_{EE'}   \sum_C \mu(C) \left(c_{jm}c_{jk}^*\right) \cdot \Re(\beta_{jE}^{(B)} \beta_{jE}^{(Y)*} \beta_{mE'}^{(A)}\beta_{kE'}^{(X)*}). \nonumber\\
         &\langle \psi_m |\cM^{(rr,ABXY)}(\rho_{\beta 0})	|\psi_k\rangle = p_{k0}\sum_{j;E,E'}  f_{E'E}   \sum_C \mu(C) \left(c_{jm}c_{jk}^*\right)  \cdot \Re\left(\beta_{jE}^{(B)} \beta_{jE}^{(Y)*} \beta_{mE'}^{(A)}\beta_{kE'}^{(X)*}\right)\nonumber
      \end{align}
      where $\Re(\alpha)=\frac{1}{2}(\alpha+\alpha^*)$ is the real part of a complex number.
\end{lemma}
\begin{proof}
Recall that  
\begin{align}
	\rho_{\beta 0} &= \sum_j p_{j0} \ket{\psi_j} \bra{\psi_j}.
	\end{align}

Suppose the gates chosen in the algorithm are $\cC=\{\QPE,C\}$. Similarly as how we derive Eq.~(\ref{eq:accept}) in Section \ref{sec:algo}, one can check that,
\begin{align}
	&W^{(10)} U_\cC^{(AB)} \ket{\psi_j}\ket{0^{2gr+1}} =  \sum_{l;E, E'} \sqrt{f_{EE'}}  \,c_{jl} \cdot\beta_{jE}^{(B)}\cdot  \beta_{lE'}^{(A)} \ket{\psi_l}\ket{E}\ket{E'}  \ket{1}.
\end{align}

   \noindent\underline{(a) {\acc~Case:}}
 
   The whole operator $\cM^{(a,ABXY)}_C$ then looks like:
    \begin{align}
   \cM^{(a,ABXY)}_{\cC} ( \ket{\psi_j}\bra{\psi_j})    =  & tr_{2,3,4} \left( W^{(10)}U^{(AB)}_\cC ~[\ket{\psi_j}\bra{\psi_j}\otimes \ket{0^{2gr+1}}\bra{0^{2gr+1}}]~ (U_\cC^{(XY)})^\dagger (W^{(10)})^\dagger \right)\\
        = &  \sum_{l,h;E, E'} f_{EE'}  (c_{jl}c_{jh}^*)  (\beta_{jE}^{(B)} \beta_{jE}^{(Y)*}) (\beta_{lE'}^{(A)}\beta_{hE'}^{(X)*}) \ket{\psi_l}\bra{\psi_h}.
    \end{align}
    Then
    \begin{align}
        \langle \psi_m| \cM^{(a,ABXY)}_\cC \left[ \rho_{\beta 0} \right]| \psi_k\rangle = \sum_j p_{j0}  \sum_{E, E'} f_{EE'}  (c_{jm}c_{jk}^*)  (\beta_{jE}^{(B)} \beta_{jE}^{(Y)*}) (\beta_{mE'}^{(A)}\beta_{kE'}^{(X)*}).\label{eq:why_FQPE_acc_1}
    \end{align}
    When we average over the choice  of $\QPE$ and $\FQPE$, the $(\beta_{jE}^{(B)} \beta_{jE}^{(Y)*})  (\beta_{mE'}^{(A)}\beta_{kE'}^{(X)*})$ term will become:
    \begin{align}
        \frac{1}{2}\left( \beta_{jE}^{(B)} \beta_{jE}^{(Y)*}  \beta_{mE'}^{(A)}\beta_{kE'}^{(X)*} + \beta_{jE}^{(B)*} \beta_{jE}^{(Y)}  \beta_{mE'}^{(A)*}\beta_{kE'}^{(X)}\right)=\Re(\beta_{jE}^{(B)} \beta_{jE}^{(Y)*} \beta_{mE'}^{(A)}\beta_{kE'}^{(X)*}).
    \end{align}
    We  also average over the choice of $C$ and finally get:
      $$\bra{ \psi_m} \cM^{(a,ABXY)}(\rho_{\beta 0})\ket{\psi_k} 
           =  \sum_j p_{j0} \sum_{E, E'} f_{EE'}   \sum_C \mu(C) \left(c_{jm}c_{jk}^*\right) \cdot \Re(\beta_{jE}^{(B)} \beta_{jE}^{(Y)*} \beta_{mE'}^{(A)}\beta_{kE'}^{(X)*}).$$
    
   \noindent\underline{(b) {\rej~Case:}} Note that 
   \begin{align}
   &\bra{0^{2gr+1}} \bra{\psi_m} (U_\cC^{(BA)})^\dagger (W^{(10)})^\dagger W^{(10)} U_\cC^{(YX)} \ket{\psi_k}\ket{0^{2gr+1}}\\  &=\sum_{l;E,E'}  f_{EE'}   \left(c_{ml}^*c_{kl}\right) \left( \beta_{lE'}^{(B)*}\beta_{lE'}^{(Y)}\right) \left( \beta_{mE}^{(A)*}\beta_{kE}^{(X)}\right)\nonumber \\
    &=\sum_{j;E,E'}  f_{E'E}   \left(c_{mj}^*c_{kj}\right) \left( \beta_{jE}^{(B)*}\beta_{jE}^{(Y)}\right) \left( \beta_{mE'}^{(A)*}\beta_{kE'}^{(X)}\right).\nonumber
   \end{align}
   where the last equality comes from
changing the name $l,E,E'$ to $j,E',E$. 
Use the definition of $\cM^{(rr,ABXY)}$, we further have
\begin{align}
\langle \psi_m |\cM^{(rr,ABXY)}_\cC(\rho_{\beta 0})	|\psi_k\rangle
	 &= p_{k0} \bra{0^{2gr+1}} \bra{\psi_m} (U_\cC^{(BA)})^\dagger (W^{(10)})^\dagger W^{(10)} U_\cC^{(YX)} \ket{\psi_k}\ket{0^{2gr+1}}\\
	 &=p_{k0}\sum_{j;E,E'}  f_{E'E}   \left(c_{mj}^*c_{kj}\right) \left( \beta_{jE}^{(B)*}\beta_{jE}^{(Y)}\right) \left( \beta_{mE'}^{(A)*}\beta_{kE'}^{(X)}\right) \label{eq:why_FQPE_rej_1}
\end{align}
Note that in Eq.~(\ref{eq:why_FQPE_rej_1}), the term $ \left( \beta_{jE}^{(B)*}\beta_{jE}^{(Y)}\right) \left( \beta_{mE'}^{(A)*}\beta_{kE'}^{(X)}\right)$ is the complex conjugate of the corresponding term in  Eq.~(\ref{eq:why_FQPE_acc_1}). This is why in our algorithm we use $\QPE$ and $\FQPE$ randomly to cancel this phase: 
 when we average over the choice  of $\QPE$ and $\FQPE$, the $\left( \beta_{jE}^{(B)*}\beta_{jE}^{(Y)}\right) \left( \beta_{mE'}^{(A)*}\beta_{kE'}^{(X)}\right)$ term will become:
    \begin{align}
        \frac{1}{2}\left( \beta_{jE}^{(B)*}\beta_{jE}^{(Y)} \beta_{mE'}^{(A)*}\beta_{kE'}^{(X)}+ \beta_{jE}^{(B)}\beta_{jE}^{(Y)*} \beta_{mE'}^{(A)}\beta_{kE'}^{(X)*}\right)=\Re(\beta_{jE}^{(B)} \beta_{jE}^{(Y)*} \beta_{mE'}^{(A)}\beta_{kE'}^{(X)*}).
    \end{align}
    We  also average over the choice of $C$ and  get:
      $$\langle \psi_m |\cM^{(rr,ABXY)}(\rho_{\beta 0})	|\psi_k\rangle = p_{k0}\sum_{j;E,E'}  f_{E'E}   \sum_C \mu(C) \left(c_{mj}^*c_{kj}\right)  \cdot \Re\left(\beta_{jE}^{(B)} \beta_{jE}^{(Y)*} \beta_{mE'}^{(A)}\beta_{kE'}^{(X)*}\right).$$
    Further note that  $\mu$ chooses $C$ and $C^\dagger$ with the same probability, thus
$$\sum_{C} \mu(C)  c_{jm}c_{jk}^* = \sum_{C} \mu(C^\dagger)  c_{mj}^*c_{kj} =  \sum_{C} \mu(C)  c_{mj}^*c_{kj}.$$
where the amplitudes for $C^{\dag}$ are obtained from $C$ by swapping the indices and taking the complex conjugate. Thus we conclude the proof.

 \end{proof}

\subsection{Uniform Error}\label{sec:fp_trun_uni}
Before we analyze the errors, we first prove a Lemma which bound the trace norm of certain matrices.

\begin{lemma}\label{lem:uvst_norm}
If $2\beta w\leq 1$, for any four subsets $A,B,X,Y \subseteq \nI$, we have
\begin{align}
\tn{	 \cM_\cC^{(a,ABXY)}(\rho_{\beta 0})  }\leq 2, \quad \quad\tn{	 \cM^{(a,ABXY)}(\rho_{\beta 0})  }\leq 2.
\end{align}
\end{lemma}

\begin{proof} 
 To ease notation, we abbreviate $\ket{\psi_j}
\ket{0^{2gr+1}}$ as $\ket{\psi_j0^{2gr+1}}$. Recall that

\begin{align}
	&	\cM_\cC^{(a,ABXY)}(\ket{\psi_j}\bra{\psi_j}) = 	 tr_{2,3,4} \left( W^{(10)}U^{(AB)}_\cC ~\ket{\psi_j0^{2gr+1}}\bra{\psi_j0^{2gr+1}}~ (U_\cC^{(XY)})^\dagger (W^{(10)})^\dagger \right) \label{eq:tn_auvst}\\
  &\|U_\cC^{(AB)}\| \leq 	 \|P^{(A)}\|\cdot  \|\sQPE_{1,3}\| \cdot \|C\| \cdot \|P^{(B)}\|\cdot \|\sQPE_{1,2}\| \leq 1.
 \end{align}
 where in the last inequality we use that $P^{(A)},P^{(B)}$ are projections thus their spectrum norm is bounded by $1$. Besides, $\|W^{(10)}\|\leq 1$ by definition. Thus by Corollary \ref{cor:unitTN} in Appendix \ref{appendix:Matrix_Norm} we  get 
 \begin{align}
     \tn{\cM_\cC^{(a,ABXY)}(\ket{\psi_j}\bra{\psi_j})}\leq 1.
 \end{align}

Then when $2\beta w\leq 1$, by Lemma \ref{lem:pq_bound_QPEr} we have  
 $p_{j0}\leq  p_{j} \cdot (1+2\beta w)\leq 2p_j$. By triangle inequality we have
 \begin{align}
\tn{\cM_\cC^{(a,ABXY)}(\rho_{\beta 0})} &\leq \sum_j 2p_j \tn{\cM_\cC^{(a,ABXY)}(\ket{\psi_j}\bra{\psi_j})}\leq 2.
\end{align}

When averaging over the random selection of $\cC$, we can get 
the bound for $\cM^{(a,uvst)}$ by triangle inequality.

\end{proof}

The following is a key lemma in the analysis which clusters the projected operators and bounds the norm of each cluster separately. While $\bra{\psi_m}(\cM^{(a,AvsY)} -\cM^{(rr,AvsY)})(\rho_{\beta0})\ket{\psi_k} $ may not cancel exactly to $0$,
the error for each cluster is independent of $m$ and $k$, which allows us to factor out the 
error term over the entire matrix.

\begin{lemma}\label{lem:uniform_error}
Recall that $T=\{0,1,``else\}$. Consider four subsets $A,B,X,Y\subseteq T$. If $2\beta w\leq 1$
	\begin{itemize}
		\item[(1)] For any $v,s\in \{0,1\}$, we have
		 \begin{align*}
	&\bra{\psi_m}(\cM^{(a,AvsY)} -\cM^{(rr,AvsY)})(\rho_{\beta0})\ket{\psi_k} = (1-e^{\beta (s-v) w })\bra{ \psi_m} \cM^{(a,AvsY)}(\rho_{\beta 0})\ket{\psi_k}.\\
 &\tn{\left(\cM^{(a,AvsY)} -\cM^{(rr,AvsY)}\right)(\rho_{\beta0})}\leq 4\beta w.
	\end{align*} 
  Note that the error $e^{\beta (s-v) w }$ in the first equality is uniform and independent of $m,k$.
		\item[(2)] If one of $A,B,X,Y$ equal to $\{``else"\}$, then   we have
	 \begin{align*}
	 	&\tn{\cM^{(a,ABXY)}(\rho_{\beta0})}\leq 2\cdot 2^{-\cgh +2n}.\\
	&\tn{\cM^{(rr,ABXY)}(\rho_{\beta0})}\leq 2\cdot 2^{-\cgh +2n}.
	 \end{align*}
	\end{itemize}
\end{lemma}

\begin{proof}
	Recall that from Lemma \ref{lem:expr_rp_QPE} we have for any $m,k$, 
	 \begin{align}
           &\bra{ \psi_m} \cM^{(a,ABXY)}(\rho_{\beta 0})\ket{\psi_k} 
           =  \sum_j p_{j0} \sum_{E, E'} f_{EE'}   \sum_C \mu(C) \left(c_{jm}c_{jk}^*\right) \cdot \Re(\beta_{jE}^{(B)} \beta_{jE}^{(Y)*} \beta_{mE'}^{(A)}\beta_{kE'}^{(X)*}). \label{eq:aABXY}\\
         &\langle \psi_m |\cM^{(rr,ABXY)}(\rho_{\beta 0})	|\psi_k\rangle = p_{k0}\sum_{j;E,E'}  f_{E'E}   \sum_C \mu(C) \left(c_{jm}c_{jk}^*\right)  \cdot \Re\left(\beta_{jE}^{(B)} \beta_{jE}^{(Y)*} \beta_{mE'}^{(A)}\beta_{kE'}^{(X)*}\right)\nonumber
      \end{align}
Recall that 
$$p_{js} = \exp{(-\beta  E_j^{(s)}})/Z.$$
     With some abuse of notation, we use $f$ for both
     \begin{align}
    &f_{EE'}=\min \left\{1, \exp\left(\beta \overline{E}- \beta \overline{E'}\right) \right\}	\\
     &f_{E_k^{(s)}E_j^{(v)}}=\min \left\{1, \exp\left(\beta E_k^{(s)}-\beta E_j^{(v)}\right) \right\}
     \end{align}
     In other words, if $E$ is a vector of $g$ energies from $S(r)$, the function $f$ implicitly takes the median value in determining $\min \left\{1, \exp\left(\beta \overline{E}- \beta \overline{E'}\right) \right\}$.
   
  \underline{For (1)}: 
     Suppose $v,s\in\{0,1\}$. The key thing to notice is that, by definition $\beta_{jE}^{(v)}$ is non-zero only if $\overline{E}=E_j^{(v)}$. Similarly for $\beta_{kE'}^{(s)*}$. Thus   $\bra{ \psi_m} \cM^{(a,AvsY)}(\rho_{\beta 0})\ket{\psi_k} $  is a sum of terms, where all the non-zero terms  $f_{EE'}$ take a uniform value as $f_{E^{(v)}_jE^{(s)}_k}$:
          \begin{align}
     \bra{ \psi_m} \cM^{(a,AvsY)}(\rho_{\beta 0})\ket{\psi_k} 
          & =  \sum_j p_{j0} \sum_{E, E':  \overline{E}=E^{(v)}_j \atop \overline{E'}=E^{(s)}_k} f_{EE'}   \sum_C \mu(C) \left(c_{jm}c_{jk}^*\right) \cdot \Re(\beta_{jE}^{(v)} \beta_{jE}^{(Y)*} \beta_{mE'}^{(A)}\beta_{kE'}^{(s)*})\nonumber	\\
           &=  \sum_j p_{j0} \cdot f_{E^{(v)}_jE^{(s)}_k} \sum_{E, E':  \overline{E}=E^{(v)}_j \atop \overline{E'}=E^{(s)}_k}     \sum_C \mu(C) \left(c_{jm}c_{jk}^*\right) \cdot \Re(\beta_{jE}^{(v)} \beta_{jE}^{(Y)*} \beta_{mE'}^{(A)}\beta_{kE'}^{(s)*})\nonumber
     \end{align}
Similarly we have  
  \begin{align}
    \langle \psi_m |\cM^{(rr,AvsY)}(\rho_{\beta 0})	|\psi_k\rangle &= p_{k0}\sum_{j} \sum_{E,E': \overline{E}=E_{j}^{(v)}\atop \overline{E'}=E_{k}^{(s)} }  f_{E'E}   \sum_C \mu(C) \left(c_{jm}c_{jk}^*\right)  \cdot \Re\left(\beta_{jE}^{(v)} \beta_{jE}^{(Y)*} \beta_{mE'}^{(A)}\beta_{kE'}^{(s)*}\right)\nonumber\\
           &= \sum_j  p_{k0}\cdot  f_{E_k^{(s)} E_j^{(v)}}  \sum_{E,E': \overline{E}=E_{j}^{(v)}\atop \overline{E'}=E_{k}^{(s)} }    \sum_C \mu(C) \left(c_{jm}c_{jk}^*\right)  \cdot \Re\left(\beta_{jE}^{(v)} \beta_{jE}^{(Y)*} \beta_{mE'}^{(A)}\beta_{kE'}^{(s)*}\right)\nonumber
     \end{align}

  Note that by definition of $f$ and $p_{ks}$ we always have
  \begin{align}
   &p_{ks}\cdot f_{E_k^{(s)}E_j^{(v)}} = p_{jv}\cdot  f_{E_j^{(v)}E_k^{(s)}}\\
   &p_{ks}=p_{k0}\cdot e^{-\beta s w}
  \end{align}
Thus 
\begin{align}
	p_{k0}\cdot  f_{E_k^{(s)} E_j^{(v)}} &=  e^{\beta s w } p_{ks}\cdot f_{E_k^{(s)}E_j^{(v)}}\\
	&=  e^{\beta s w } p_{jv}\cdot  f_{E_j^{(v)}E_k^{(s)}}\\
	&= e^{\beta (s-v) w } p_{j0}\cdot f_{E_j^{(v)}E_k^{(s)}}
	\end{align}
Note that the ``error" $e^{\beta (s-v) w }$ is independent of $j,m,k$,
thus
\begin{align}
	 &\bra{ \psi_m}\left( \cM^{(a,AvsY)}-\cM^{(rr,AvsY)}\right) (\rho_{\beta 0})	|\psi_k\rangle \\ &= (1-e^{\beta (s-v) w })  \sum_j p_{j0} \cdot f_{E^{(v)}_jE^{(s)}_k} \sum_{E, E':  \overline{E} =E^{(v)}_j \atop \overline{E'}=E^{(s)}_k}     \sum_C \mu(C) \left(c_{jm}c_{jk}^*\right) \cdot \Re(\beta_{jE}^{(v)} \beta_{jE}^{(Y)*} \beta_{mE'}^{(A)}\beta_{kE'}^{(s)*})\\
	 &=(1-e^{\beta (s-v) w })\bra{ \psi_m} \cM^{(a,AvsY)}(\rho_{\beta 0})\ket{\psi_k}
\end{align}
 Thus
\begin{align}
     &\left(\cM^{(a,AvsY)}-\cM^{(rr,AvsY)}\right)(\rho_{\beta0}) =(1-e^{\beta (s-v) w }) \cdot \cM^{(a,AvsY)}(\rho_{\beta 0}).\\
     &\tn{\left(\cM^{(a,AvsY)}-\cM^{(rr,AvsY)}\right)(\rho_{\beta0})} \leq 4\beta w. 
\end{align}
where for the last inequality, we use Lemma \ref{lem:uvst_norm} and the fact that  since $v,s\in\{0,1\}$, we have
$|s-v|\leq 1$ and  $|1-e^{\beta(s-v)w}| \leq 2\beta w$. 

\underline{For (2).}
W.o.l.g. assume $X=\{``else"\}$. Other cases are similar.
 Note  that
\begin{align}
	&\left(\sum_j \left|c_{jm}c_{jk}^*\right| \right)^2 \leq  \left(\sum_j \left|c_{jm}\right|^2 \right)\left(\sum_j \left|c_{jk}^*\right|^2 \right)\leq 1.\label{eq:c_bound}\\
	&\left(\sum_{E} |\beta_{jE}^{(B)} \beta_{jE}^{(Y)*}|\right)^2 \leq  \left(\sum_E \left|\beta_{jE}^{(B)}\right|^2 \right)\left(\sum_E \left|\beta_{jE}^{(Y)*}\right|^2 \right) \leq 1\cdot 1\\
	&\left(\sum_{E'} |\beta_{mE'}^{(A)} \beta_{kE'}^{(X)*}|\right)^2 \leq  \left(\sum_{E'} \left|\beta_{mE'}^{(A)}\right|^2 \right)\left(\sum_{E'} \left|\beta_{kE'}^{(X)*}\right|^2 \right) \leq 2^{- \cg  }\cdot 1.
\end{align}
where the first equality comes from the fact that   $C$ is a unitary. The second equality comes from  $\{\beta_{jE}^{(B)}\}_E$ is a subset of $\{\beta_{jE}\}_E$, and $\{\beta_{jE}\}_E$ is the amplitude of a quantum state. For the third equality, recall that $X=\{``else"\}$, notice that
for any non-zero $\beta_{mE'}^{else}$, by definition we have $\overline{E'}\not\in \{ E_j^{(0)},E_j^{(1)}\}$.
Then the third equality comes from property of $\QPE$, that is Lemma \ref{lem:QPEerr}.

Besides, note that 
\begin{align}
	|\Re(\beta_{jE}^{(B)} \beta_{jE}^{(Y)*} \beta_{mE'}^{(A)}\beta_{kE'}^{(X)*})| \leq |\beta_{jE}^{(B)} \beta_{jE}^{(Y)*} \beta_{mE'}^{(A)}\beta_{kE'}^{(X)*}|.
\end{align}
Since $2\beta w\leq 1$, by Lemma \ref{lem:pq_bound_QPEr} we have $ p_{j0} \in [0,2]$. Note that $f_{EE'}\in [0,1]$ and $\sum_C \mu(C)=1$, then from Eq.~(\ref{eq:aABXY})  we have for any $m,k$,
\begin{align}
|\bra{ \psi_m} \cM^{(a,ABXY)}(\rho_{\beta 0})\ket{\psi_k}| 
           &\leq  2    \sum_C \mu(C)  \sum_j \left|c_{jm}c_{jk}^*\right|  \cdot  \sum_{E} |\beta_{jE}^{(B)} \beta_{jE}^{(Y)*}| \cdot  \sum_{E'} |\beta_{mE'}^{(A)}\beta_{kE'}^{(X)*}| \\
           &       \leq 2\cdot 2^{-\cgh }.
 \end{align}  
Then by triangle inequality,
\begin{align}
    \tn{\cM^{(a,ABXY)}(\rho_{\beta0})} \leq \sum_{m,k} |\langle \psi_m|\cM^{(a,ABXY)}(\rho_{\beta0}) \left|\psi_k\rangle \right| \cdot \tn{\,\ket{\psi_m}\bra{\psi_k}\,} \leq   2\cdot 2^{-\cgh +2n}
\end{align}

 The proof for $\tn{\cM^{(rr,ABXY)}(\rho_{\beta0})}$ is similar.

\end{proof}

\subsection{Gibbs state as Approximate Fixed Point}\label{sec:fp_trun_fin}

We are ready now to complete the proof of Lemma \ref{lem:fixed_point}, which provides an upper bound on $\tn{\cLt(\rho_{\beta})}$.

\begin{proof}[of Lemma \ref{lem:fixed_point}]
If $r\geq r_{\beta H}$, we have $2\beta w\leq 1$.
From Lemma \ref{lem:fp_trun} and Lemma \ref{lem:rl2rr} we have
\begin{align}
	\tn{\cLt(\rho_{\beta})}
	\leq  2  \tn{ \cM^{(a)}(\rho_{\beta 0})-   \cM^{(rr)}(\rho_{\beta 0})} + 4\cdot 2\beta w
\end{align}  
Note that for any matrices $\{N^{(uvst)}\}_{uvst}$, by triangle inequality we have
\begin{align*}
     \tn{\sum_{u,v,s,t\in \nI}N^{(uvst)}} \leq 
     \tn{ \sum_{\substack{  v=else\\ u,s,t\in T}}N^{(uvst)}} + 
     \tn{ \sum_{\substack{  v\in\{0,1\}, s=else \\ u,t\in T}}N^{(uvst)}} 
    + \sum_{v,s\in\{0,1\}} \tn{ \sum_{u,t\in \nI}N^{(uvst)}}
\end{align*}

Then by Lemma \ref{lem:eq_uvst} and  Lemma \ref{lem:uniform_error}, we have 
\begin{align}
	\tn{ \cM^{(a)}(\rho_{\beta 0})-   \cM^{(rr)}(\rho_{\beta 0})} \leq 2\cdot( 2\cdot 2^{-\cgh +2n}+ 2\cdot 2^{-\cgh +2n}) + 4 \cdot 4\beta w\cdot
\end{align}
Substituting $w=\kappa_H \cdot 2^{-r}$,  we complete the proof.
\end{proof}

\section{Proofs of Theorem \ref{thm:main} and Theorem \ref{thm:mix}}\label{sec:eb_final}

\begin{proof}[of Theorem \ref{thm:main}]
The fact that one step of the algorithm can be expressed as $\cET(\rho) =\left( \cI + \tau^2 \cLt  + \tau^4 \cJt[\tau]  \right)(\rho)$, where $\norm{\cJt[\tau]}\leq 4$ is proven in Lemma \ref{lem:evol_main} in Section \ref{sec:operator_LC}.
 The bound on $\tn{\cLt(\rho_{\beta})}$ is proved in Lemma \ref{lem:fixed_point} in Section \ref{sec:fp}. The Uniqueness and Relaxation property is proved in Theorem  \ref{thm:uniqueness} in Section \ref{sec:Lind_unique}.
\end{proof}

\vspace{.1in}
Before proving theorem \ref{thm:mix} we  give a bound for the evolution. 
\begin{lemma}[Evolution]\label{lem:evol_new}
    For any $t\in \bR$, $\tau\in (0,1]$. If $K=t/\tau^2$ is an integer, then
    \begin{align}
        \norm{\cEN^K -e^{t\cLt}} \leq 2 e^4 K\tau^4 .
    \end{align}
\end{lemma}
\begin{proof}
    Note that
    \begin{align}
        \cEN^K -(e^{\tau^2\cLt})^K = \sum_{k=0}^{K-1}  (e^{\tau^2\cLt})^k \left( \cEN- e^{\tau^2\cLt} \right) \cEN^{K-k-1}.
    \end{align}
   Note that $\cEN$ is CPTP by definition.  By Lemma \ref{lem:eLCPTP} we know $e^{\tau^2\cLt}$ is also CPTP. Thus  $\norm{\cEN}$ and $\norm{e^{\tau^2\cLt}}$ are bounded by $1$ by Lemma \ref{lem:CPTP} in Appendix \ref{appendix:Matrix_Norm}. Then we have 
   \begin{align}
       \norm{\cEN^K -(e^{\tau^2\cLt})^K}\leq K \cdot 1\cdot \norm{\cEN- e^{\tau^2\cLt}} \cdot 1\leq K \cdot \tau^4 2e^4 .
   \end{align}
   where for the last inequality, recall that  $\cEN=\cI+\tau^2 \cLt + \tau^4 \cJt[\tau]$ and the Taylor expansion of $e^{\tau^2\cLt}$, we get 
   \begin{align}
       \norm{\cEN- e^{\tau^2\cLt}} \leq \norm{\cJt[\tau]}\tau^4 +  \sum_{k=2}^\infty \frac{\tau^{2k} \norm{\cLt}^k }{k!} \leq \tau^4 2e^4.
   \end{align} 
   where we use $\norm{\cJt}\leq 4$ and $\norm{\cLt}\leq 4$ from Lemma \ref{lem:evol_main}.   
\end{proof}

We then prove that the distance  between $\rho_\beta$ and $\rho_\cLt$ can be bounded in terms of the mixing time.

\begin{lemma}\label{lem:dist_L_beta}
    $$\tn{\rho_\cLt-\rho_\beta}\leq \epsilon +  \tn{\cLt(\rho_\beta)}\cdot t_{mix}(\cLt,\epsilon) $$
\end{lemma}
\begin{proof}
    We abbreviate $t_{mix}(\cLt,\epsilon)$ as $t$. Let $\tau>0$ be a parameter such that $K=t/\tau^2$ is an integer. 
We have 
    \begin{align}
        \tn{\rho_\cLt-\rho_\beta} \leq \tn{\rho_\cLt - e^{t\cLt}(\rho_\beta)} + \tn{ e^{t\cLt}(\rho_\beta)-\rho_\beta} \leq \epsilon + \tn{ e^{t\cLt}(\rho_\beta)-\rho_\beta}.
    \end{align}
   where the last inequality comes from the definition of mixing time. Besides,  
   from Lemma \ref{lem:evol_main} we have 
  $\norm{\cLt}\leq 4$. Expand $e^{\tau^2 \cLt}$ as Taylor series in $\tau$, and use triangle inequality we get   
  \begin{align}
    	\tn{\left( e^{\tau^2\cLt}-\cI  \right)(\rho_\beta)} &\leq \tau^2 \tn{\cLt(\rho_\beta)} + \tau^4 \sum_{k=2}^{\infty} \frac{\norm{\cLt}^k}{k!}\\
    	&\leq \tau^2 \tn{\cLt(\rho_\beta)} + \tau^4 e^4.
  \end{align}
    Thus  \begin{align}
  \tn{e^{K\tau^2\cLt}(\rho_\beta) - \rho_\beta} &\leq  \sum_{k=K}^{1} \tn{	e^{(k-1)\tau^2\cLt} \left( e^{\tau^2\cLt}-\cI  \right)(\rho_\beta)}\\
  &\leq  \sum_{k=K}^{1} \tn{	\left( e^{\tau^2\cLt}-\cI  \right)(\rho_\beta)}\\
  &\leq K\tau^2 \tn{\cLt(\rho_\beta)} + K\tau^4 e^4.
     \end{align}
  where the second inequality comes from $e^{(K-1)\tau^2\cLt}$ is CPTP by Lemma \ref{lem:eLCPTP}, thus its $\norm{\cdot}$ is bounded by $1$ by Lemma \ref{lem:CPTP} in Appendix \ref{appendix:Matrix_Norm}. Thus we  finish the proof by substituting $K\tau^2$ with the mixing time $t$ and take the limit $\tau\rightarrow 0$.
\end{proof}

\begin{proof}[of Theorem \ref{thm:mix}] 
In the proof we abbreviate  $ t_{mix}(\cLt,\epsilon)$ as  $t$ and 
 $K=t/\tau^2$.
We will analyze the following terms
\begin{align}
  \tn{\cEN^K(\rho) -\rho_\beta} \leq \tn{\cEN^K(\rho) -e^{t\cLt}(\rho)} + 
    \tn{e^{t\cLt}(\rho) -\rho_\cLt} + \tn{\rho_\cLt-\rho_\beta}.
\end{align}
By Lemma \ref{lem:evol_new} we can bound the first term by $2e^4K\tau^4$. 
    By the definition of mixing time, we can bound the second term by $\epsilon$. The bound for last term comes from Lemma \ref{lem:dist_L_beta} and Theorem \ref{thm:main}.  Thus we  get the desired error bounds by substituting $K\tau^2$ with the mixing time $t$.

    \end{proof}

\section{Acknowledgement} Part of this work was conducted while the author was visiting the Simons Institute for the Theory of Computing. We thank Yu Tong for the helpful discussion on analyzing mixing time by spectral gap of $\cLs$.

\appendix

\section{More Details on Quantum Phase Estimation}\label{appendix:QPE}

\textbf{(1) Adapting quantum phase estimation to estimating Hamiltonian eigenvalues}. Firstly we recall the
 standard quantum phase estimation in Section 5.2 in \cite{nielsen2010quantum}: 
Suppose a unitary $V$ has an eigenstate $\ket{\phi}$ with eigenvalue $e^{2\pi \mu}$, where $\mu\in [0,1)$ is unknown. Quantum Phase estimation is a quantum algorithm which has access to $\ket{\phi}$ and $V$, and outputs an estimate of $\mu$.

More precisely, let $\ket{\phi}:=\ket{\psi_j}$, $U:=e^{iHt}$, $t:= \frac{2\pi}{\kappa_H}$,  where $\kappa_H=poly(n)$ is a power of two that upper bounds  
$\|H\|$.  For example, for local Hamiltonian $H=\sum_{i=1}^{m} H_i, \|H_i\|\leq 1$, one can set $\kappa_H$ to be the least integer which is a power of two and is greater than $m$.
Then $\ket{\psi_j}$ is an eigenstate of $U$ with eigenvalue $e^{2\pi \mu}$, for 
$$\mu := \frac{E_j t}{2\pi}=\frac{E_j}{\kappa_H}\in[0,1).$$

 Let $r$ be an integer. For any $\bb\in\{0,1\}^r$, with some abuse of notations, we use $\bb$ both for the binary string and the integer $\sum_{j=1}^r b_j2^{j-1}$.
 Denote $\bb^{(j)}\in\{0,1\}^r$ be the integer such that $\bb^{(j)}/2^r$ is the best $r$ bit approximation to $\frac{E_j}{\kappa_H}$ which is less than $\frac{E_j}{\kappa_H}$.   
 Then if $\bb$  is a good approximation to $\bb^{(j)}$, we have $E(\bb)$ defined below is a good approximation of $E_j$:
 \begin{align}\label{eq:Ebb}
 E(\bb):=  \kappa_H \cdot \bb/2^r.
\end{align}
Note that since $\kappa_H$ is a power of $2$, the binary representation of $\bb$
 can also be viewed as the binary representation of $E(\bb)$ by shifting the decimal point by $\log (\kappa_H)$ bits.

  The quantum phase estimation w.r.t precision $r$ in \cite{nielsen2010quantum} is a unitary, denoted as $\qpe$, which   outputs a distribution highly peaked at $\bb^{(j)}$,
  \begin{align}
	&\qpe	\ket{\psi_j}\ket{0^r} =\ket{\psi_j}\sum_{ l=-2^{r-1}+1 }^{2^{r-1}} \gamma_{jl} \ket{\bb^{(j)}+l\Mod 2^r};\label{eq:gammajl}\\
 \text{where }  &\gamma_{jl} := \frac{1}{2^r} \sum_{k=0}^{2^r-1} \left(e^{2\pi i (\epsilon_j -l/2^r)} \right)^k
 = \frac{1}{2^r} \left( \frac{1-e^{2\pi i (2^{r}\epsilon_j -l)}}{1-e^{2\pi i (\epsilon_j -l/2^{r})}} \right).\\
&\epsilon_j := E_j/\kappa_H- b^{(j)}/2^r.
	\end{align} 

\begin{lemma} Note that
    \begin{align}
        |\gamma_{j0}|^2 + |\gamma_{j1}|^2 \geq 0.8.\label{eq:QPE_err_flo_cei}
    \end{align}
\end{lemma}
 \begin{proof}
 By definition of $\bb^{(j)}$ and $\epsilon_j$ we have that $\epsilon_j\in [0,2^{-r})$.
  Note that if $\epsilon_j = 0$, then $ |\gamma_{j0}|^2 = 1$, so for the remainder of the proof, we will assume $\epsilon_j \in (0,2^{-r})$. 
 Define $\theta= 2^r \epsilon_j\in(0,1)$.
     Note that
     \begin{align}
         \gamma_{jl} = \frac{1}{2^r} \left( \frac{1-e^{2\pi i (\theta -l)}}{1-e^{2\pi i (\theta -l)/2^{r}}} \right) = \frac{1}{2^r} \frac{e^{\pi i(\theta -l)}}{e^{\pi i(\theta -l)/2^r}} \cdot \frac{\sin(\pi(\theta -l))}{\sin(\pi(\theta -l)/2^r)}.
     \end{align}
So the goal is to lower bound:
\begin{align}
    |\gamma_{j0}|^2 + |\gamma_{j1}|^2  = \frac{1}{4^r} \frac{\sin^2 (\theta \pi)}{\sin^2(\theta \pi/2^r)} + \frac{1}{4^r} \frac{\sin^2 ((\theta-1) \pi)}{\sin^2((\theta-1) \pi/2^r)}.
\end{align}
Since $\sin(\theta) \le \theta$ for $\theta\geq 0$, we can use this lower bound in the denominator for each term:
\begin{align}
    |\gamma_{j0}|^2 + |\gamma_{j1}|^2 &\ge \frac{1}{4^r} \frac{\sin^2 (\theta \pi)}{(\theta \pi/2^r)^2} + \frac{1}{4^r} \frac{\sin^2 ((\theta-1) \pi)}{((\theta-1) \pi/2^r)^2} \\
    &= 
\frac{1}{\pi^2} \left( \frac{\sin^2 (\theta \pi)}{\theta^2} +  \frac{\sin^2 (\theta-1) \pi)}{(\theta-1)^2} \right).\label{eq:function}
\end{align}
Note that as shown in Figure \ref{fig:function} for $\theta \in (0,1)$, the function in Eq.~(\ref{eq:function}) is symmetric around $\theta = 1/2$ and is minimized for $\theta = 1/2$.
The value obtained by plugging $\theta = 1/2$ into Eq.~(\ref{eq:function}) is at least  0.8 . \begin{figure}[H] \centering
\includegraphics[width=0.3\textwidth]{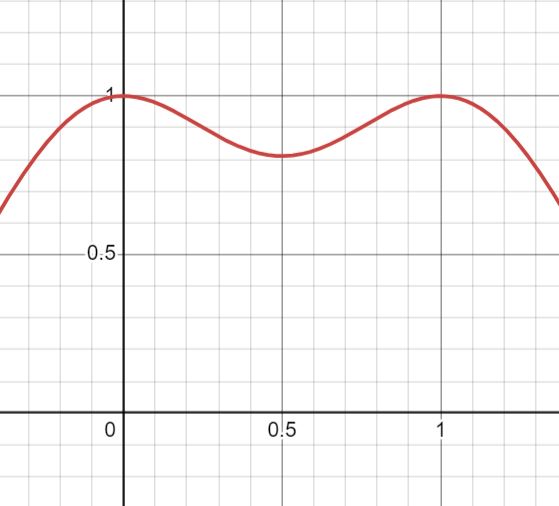}\caption{The function shown in Equation (\ref{eq:function}) is minimized for $\theta = 1/2$.} \label{fig:function}
\end{figure}
 \end{proof}

Note that  $\gamma_{j0}$ and $\gamma_{j1}$ are the amplitudes of $\bb^{(j)}$ and $\bb^{(j)}+1$, that are the best two $r$-bit approximation to $\frac{E_j}{\kappa_H}$. Denote the best two $r$-bit approximation to $E_j$  as 
\begin{align}
    &\flo{E_j}:= E(\bb^{(j)}), \quad \cei{E_j}:= E(\bb^{(j)}+1).
\end{align}
Since there is a one-to-one correspondence between $r$-bit string $\bb$ and $E(\bb)\in S(r)$, to ease notations, we  rewrite  Eqs.~(\ref{eq:gammajl})(\ref{eq:QPE_err_flo_cei}) as

\begin{align}
	&\qpe	\ket{\psi_j}\ket{0^r} =\ket{\psi_j}\sum_{\nu\in S(r)} \beta_{j\nu} \ket{\nu},\\
	&\text{where } |\beta_{j\flo{E_j}}|^2 + |\beta_{j\cei{E_j}}|^2\geq 0.8. \label{eq:QPE_tricky}
	\end{align}

 In addition, one can check that $\forall j,$
 \begin{align}
     \sum_{\nu\in S(r)}\beta_{j\nu} &= \sum_{l=-2^{r-1}+1}^{2^{r-1}} \gamma_{jl}\\
     &= \frac{1}{2^r} \sum_{k=0}^{2^r-1}  \sum_{l=-2^{r-1}+1}^{2^{r-1}} \left(e^{2\pi i \epsilon_j} \right)^k \left(e^{2\pi i ( -l/2^r)} \right)^k\\
     &= \frac{1}{2^r} \sum_{k=0}^{2^r-1}  \left(e^{2\pi i \epsilon_j} \right)^k \cdot 2^r \delta_{k0}\\
     &=1. \label{eq:bjE1_old}
 \end{align}

\noindent\textbf{(2) Boosted Quantum Phase Estimation.}
Let $r,g$ be two integer parameters. Taking input as $\ket{\psi_j}\ket{0^{rg}}$, the
 Boosted quantum phase estimation ($\QPE$)  in this manuscript is just viewing the $\ket{0^{rg}}$ as $g$ copies of $\ket{0^r}$, and sequentially perform the standard Quantum Phase Estimation  $\qpe$  w.r.t. $\ket{\psi_j}$ and each copy of  $\ket{0^r}$,
 \begin{align}
     \QPE \ket{\psi_j}\ket{0^{gr}} &= \ket{\psi_j}\sum_{\nu_1,...,\nu_g\in S(r)} \beta_{j\nu_1}...\beta_{j\nu_g} \ket{\nu_1...\nu_g}.\\
     &= \ket{\psi_j}\sum_{E\in S(r)^{\otimes g}} \beta_{jE} \ket{E}.
 \end{align}
 where in the last equality to ease notations we write $E:=\nu_1...\nu_g$ and  $\beta_{jE} := \beta_{j\nu_1}...\beta_{j\nu_g}$.
We use $\overline{E}$ to denote the median of $\nu_1,...,\nu_g$. 
By Eq.~(\ref{eq:QPE_tricky}) and Chernoff bound, we have 
\begin{align}
    &\sum_{E\in S(r)^{\otimes g}: \atop \overline{E}\neq \flo{E_j}, \overline{E}\neq\cei{E_j}} |\beta_{jE}|^2  \leq  2^{- \cg }.
\end{align}
Besides, by Eq.~(\ref{eq:bjE1_old}) we have 
\begin{align}
    \sum_{E\in S(r)^{\otimes g}} \beta_{jE} =  \left(\sum_{\nu\in S(r)}\beta_{j\nu}\right)^g =1. \label{eq:bjE1}
\end{align}
\textbf{(3) Flipped Boosted Quantum Phase Estimation.}
As explained in Section 5.2 of \cite{nielsen2010quantum}, the circuit for the standard Quantum Phase Estimation $\qpe$ is, (a) First applying a layer of Hadamard gates. (b) Then apply a sequence controlled-$U$ operators, with $U$ raised to successive powers of $2$. Denote the operator as $CU$. (c) Then apply inverse quantum Fourier transform.  

The Flipped Quantum Phase Estimation $\fqpe$ is similar to $\qpe$, with the difference that each time $\fqpe$ adds a conjugate phase: (a) First applying a layer of Hadamard gate. (b) Then apply $(CU)^\dagger$. (c) Then apply  quantum Fourier transform rather than inverse quantum Fourier transform. 
It is worth noting that $\fqpe \neq \qpe^\dagger$.

Recall that from (2) the  Boosted Quantum Phase estimation is just running $\qpe$ for $g$ times. 
The Flipped Boosted Quantum Phase estimation ($\FQPE$) is just running $\fqpe$ for $g$ times.

\section{Matrix Norm Properties}\label{appendix:Matrix_Norm}
In this section we list properties on matrix norms.

\begin{lemma}[Cauchy-Schwarz inequality ]\label{lem:CS}
For any $M,N\in \Xi(m)$,
\begin{align}
 |tr(A^\dagger B)|\leq \sqrt{tr(A^\dagger A)} \sqrt{tr(B^\dagger B)}.	
\end{align}
\end{lemma}

Lemma \ref{lem:CS} is a well-known fact thus we skip the proof.

\begin{lemma}[Variational Characterization of Trace Norm]\label{lem:vcTN}
For $M\in \Xi(m)$, the following variational characterization of the trace norm holds
\begin{align}
\tn{M} = \max_{U} |tr(MU)|,	
\end{align}
 where the maximization is over all $m$-qubit unitary operators $U$.	
\end{lemma}
\begin{proof}
    Consider singular value decomposition of $M$ as $WDV$, with $W,V$  unitaries and $D$ a diagonal matrix of singular values. Using Cauchy-Schwart inequality Lemma \ref{lem:CS} one gets
    \begin{align}
        |tr(MU)| &= |tr(\sqrt{D} \sqrt{D}VUW)| \nonumber\\ 
        &\leq \sqrt{tr(\sqrt{D}\sqrt{D})}\sqrt{tr\left\{\left(\sqrt{D}VUW\right)^\dagger\left(\sqrt{D}VUW\right)\right\}}\nonumber\\
        &=tr(D).\nonumber
    \end{align}
 Note that  $tr(D)=\tn{M}$.  On the other hand, $\tn{M}$ can be obtained by choosing $U=V^\dagger W^\dagger.$
\end{proof}

\begin{corollary}\label{cor:unitTN}
Let $\ket{\phi},\ket{\varphi}$ be two unit vectors on registers $a,b$. Let $P,Q$ be linear operators on registers $a,b$ with spectrum norm bounded by $1$. Then
\begin{align}
	\tn{tr_{b}(P\ket{\phi}\bra{\varphi}Q)} \leq 1.
\end{align}
In particular, let $P=Q=I$ we get
\begin{align}
	\tn{tr_{b}(\ket{\phi}\bra{\varphi})} \leq 1.
\end{align}
\end{corollary}
\begin{proof}
	By Lemma \ref{lem:vcTN}, we have that 
	\begin{align}
\tn{tr_{b}(P\ket{\phi}\bra{\varphi}Q)} &= \max_{\text{unitary }  U \text{ on register $a$}} |tr_a\left(tr_{b}\left(P\ket{\phi}\bra{\varphi}Q \right) U\right)|\\
&= |tr_{a,b} \left(P\ket{\phi}\bra{\varphi}Q  U\right)|\\
&=|\bra{\varphi}QUP\ket{\phi}|\\
&\leq 1.
\end{align}
\end{proof}

\begin{lemma}\label{lem:unit}
	Let $\cEN:\Xi(m)\rightarrow \Xi(m)$ be a linear map. If for any unit vector $\ket{u},\ket{v}$, we have $\tn{\cEN(\ket{u}\bra{v})}\leq c$. Then $\norm{\cEN}\leq c$.
\end{lemma}
\begin{proof}
	Given any $N\in \Xi(m)$ with $\tn{N}\leq 1$, consider its spectrum decomposition 
	\begin{align}
		N=\sum_i \gamma_i \ket{u_i}\bra{v_i}.	
	\end{align}
where $\gamma_i\geq 0$ are singular values. $\ket{u_i},\ket{v_i}$ are unit vectors. $\sum_i \gamma_i \leq\tn{N}=1$. 
Then by triangle inequality we have that
\begin{align}
 \tn{\cEN(N)} \leq \sum_i \gamma_i \cdot 	\tn{\cEN(\ket{u}\bra{v})} \leq c
\end{align}
\end{proof}

\begin{lemma}\label{lem:normJ_new}
    Let $P,Q$ be operators on register $1,2,3,4$ with spectrum norm bounded by $1$. Let $\ket{w}$ be a unit vector on register $2,3,4$. For any operator $N$ on register $1$, define 
    \begin{align}
        \cF(N):= tr_{2,3,4}(P ~ [N\otimes \ket{w}\bra{w}]~ Q).
    \end{align}
    Then 
        $\norm{\cF}\leq 1.$
\end{lemma}
\begin{proof}
    By Lemma \ref{lem:unit} it suffices to assume $N=\ket{u}\bra{v}$ for unit vectors $\ket{u},\ket{v}$ and prove $\tn{\cF(N)}\leq 1$.
    Since $\ket{u}\ket{w}$ and $\ket{v}\ket{w}$  are unit vectors on register $1,2,3,4$. By Corollary \ref{cor:unitTN} we have $\norm{\cF}\leq 1$.
\end{proof}

\begin{lemma}	\label{lem:CPTP}
[Trace-Norm non-increasing of CPTP]
	Let $\cEN:\cH(m)\rightarrow \cH(m)$ be a CPTP map. We have $\norm{\cEN}\leq 1$ .That is for any Hermitian operator $M\in \cH(m)$, $\tn{\cEN(M)}\leq \tn{M} $.
\end{lemma}
\begin{proof}
	Consider the spectrum decomposition of $M=\sum_j \lambda_j\ket{\phi_j}\bra{\phi_j}$. We have that
	\begin{align}
		\tn{\cEN(M)} \leq \sum_j |\lambda_j| \cdot \tn{\cEN(\ket{\phi_j}\bra{\phi_j})}
		=\sum_j |\lambda_j| 
		=\tn{M}.
		\end{align}
where the first equality comes from the fact that $\cEN(\ket{\phi_j}\bra{\phi_j})$ is a quantum mixed state thus its trace norm is equal to $1$.
\end{proof}

\begin{lemma}[Connecting trace norm and $\sigma$-norm]\label{lem:sig_tn}
Let $\sigma$ be a Hermitian matrix where $\sigma\geq I$. Then  $\forall M \in \Xi(n)$, we have that
       $ \tn{M}^2 \leq   2^n\sn{M}{\sigma}^2.$
\end{lemma}
\begin{proof}
Note that for any two positive Semi-definite Hermitian matrices $A\geq 0,B\geq 0$, we have $tr(AB)\geq 0$. Further note that $\sigma^{1/2}-I\geq 0$, $M^\dagger \sigma^{1/2} M \geq 0$, and both of them are Hermitian. Thus we have 
\begin{align}
    \sn{M}{\sigma}^2 &=  tr(I M^\dagger I M) + tr(I M^\dagger \left(\sigma^{1/2}-I\right) M) +  tr((\sigma^{1/2}-I) \cdot M^\dagger \sigma^{1/2} M)\\
    &\geq  tr(I M^\dagger I M).
\end{align}
Denote the singular values of $M$ as $\{a_j\}_j$. (By definition $a_j\geq 0$). By Cauchy inequality 
\begin{align}
 \left( \sum_j a_j \right)^2 \leq  2^n \sum_j a_j^2 \quad
   \Rightarrow  \quad \tn{M}^2 \leq 2^n  tr(I M^\dagger I M)\leq  2^n \sn{M}{\sigma}^2.
\end{align}
\end{proof}

\section{Bounding Mixing time w.r.t. spectral gap of $\cLs$}
\label{appendix:mixing}

This section is based on \cite{wolf2012quantum} and private communication with Yu Tong. This section is primarily for the  $\cLt$ in our algorithm, while here we write a slightly more general proof for  Lindbladian $\cLt$ satisfying the following assumptions:
\begin{assumption}\label{ass:L} We assume  $e^{t\cLt}$ is CPTP for $t\geq 0$ and satisfies
    \begin{itemize}
        \item $\cLt(\rho) = \sum_{j} L_j \rho L_j^\dagger -\frac{1}{2}\left\{L_j^\dagger L_j,\rho\right\}_+$ for some matrices $\{L_j\}_j\subseteq \Xi(n)$.
        \item $\{L_j\}_j$ generates the full algebra $\Xi(n)$. There exists a unique $\rho_\cLt$ satisfying $\cLt(\rho_\cLt)=0$. Besides, this  $\rho_\cLt$ is a full-rank quantum state.
    \end{itemize}
\end{assumption}
 We use $\sigma$ to denote $\rho_\cLt^{-1}$.
 Note that $\cLt$ in our algorithm satisfies Assumption \ref{ass:L} by Lemma \ref{lem:Lind_form} and proofs in Section \ref{sec:uniqueness}.

\textbf{Notations.}
Here we define two different inner products on  $\Xi(n)$, that is the Schmidt-Hilbert inner product and the weighted inner product w.r.t. $\sigma$: For any $M,N\in \Xi(n)$,
\begin{align}
	&\langle	M,N\rangle := tr\left( M^\dagger N\right),\\
	&\langle	M,N\rangle_\sigma := tr\left( \sigma^{1/2} M^\dagger \sigma^{1/2} N\right).
\end{align}

We use $\cLd$ to denote the dual map w.r.t. Schmidt inner product. We use $\cLt^*$ to denote the dual map w.r.t. $\langle,\rangle_\sigma$. That is
\begin{align}
    &\langle M,\cLd(N)\rangle = \langle \cLt(M),N\rangle,\\   &\langle M,\cLt^*(N)\rangle_\sigma = \langle \cLt(M),N\rangle_\sigma. 
\end{align}
$\langle,\rangle_\sigma$ induces a norm where 
 $\sn{N}{\sigma}^2:=\langle N,N\rangle_\sigma$.
One can check that for any $N\in \Xi(n)$,
\begin{align}
 &\cLt^* (N)	= \sigma^{-\frac{1}{2}} \cLd \left( \sigma^{\frac{1}{2}}N \sigma^{\frac{1}{2}}\right)\sigma^{-\frac{1}{2}}\\
 &\cLd(N) =  \sum_{j} L_j^\dagger N L_j  -\frac{1}{2}\left\{L_j^\dagger L_j,N\right\}_+.
\end{align}
Define the symmetrized version of $\cLt$ as 
$$
\cLs:=\frac{\cLt+\cLt^*}{2}.
$$ By definition $\cLs$ is Hermitian w.r.t. $\langle	,\rangle_\sigma$, thus is diagonalizable and has a real spectrum.

\subsection{Properties on $\cLt$ and $\cLs$.}

In this section, we  prove some basic properties of $\cLt$ and $\cLs$. We  use the following theorem.

\begin{theorem}[Generators for Semigroup of quantum channels,  Theorem 7.1 in  \cite{wolf2012quantum}]\label{thm:CPTP}
	Consider a linear map $\cP:\Xi(n)\rightarrow \Xi(n)$. Then for any $t\geq 0$, $e^{t\cP}$ is CPTP if there
	exists a set of matrices $\{P_j\in\Xi(n)\}_j$ and a Hermitian $H_{system}$ such that
	\begin{align}
		\cP(\rho) = i [\rho,H_{system}] + \sum_{j} P_j \rho P_j^\dagger -\frac{1}{2}\left\{P_j^\dagger P_j,\rho\right\}_+.\label{eq:Lstandard}
	\end{align}
\end{theorem}

\begin{theorem} \label{thm:assumptions}
 For any $t\geq 0$,
	\begin{enumerate}
		\item[(\romannumeral1)] $\cLt(\rho_\cLt)=0$, $e^{t\cLt}$ is CPTP and  $\cLd(I)=0$.
	\item[(\romannumeral2)]  $e^{t\cLt^*} $ and $e^{t\cLs}$ are CPTP.
			\item[(\romannumeral3)] $\cLs$ has a unique fixed point, that is $\rho_\cLt$. 
		\item[(\romannumeral4)] The spectrum of $\cLs$ lies in $[-\infty,0]$.
		\item[(\romannumeral5)] The spectral gap of $\cLs$ is strictly greater than $0$, which is denoted as $\Upsilon$ and equals to
		\begin{align}
		\Upsilon (\cLs) = \min_{\langle N,N\rangle_\sigma =1 \atop \langle N,\rho_\cLt\rangle_\sigma=0} -\langle N,\cLs(N)\rangle_\sigma >0. \label{eq:specGap}
		\end{align}
	\end{enumerate}
 Since $\cLs$ is clear in the context we abbreviate $\Upsilon (\cLs)$  as $\Upsilon$.
\end{theorem}

\begin{proof} For (\romannumeral1):    Besides, by Assumption \ref{ass:L} we have $\cLt(\rho_\cLt)=0$. Besides
\begin{align}
		&\cLt(N) =  \sum_{j} L_j N L_j^\dagger -\frac{1}{2}\left\{L_j^\dagger L_j,N\right\}_+.\label{eq:Lss}\\
            &\cLd(N) =  \sum_{j} L_j^\dagger N L_j  -\frac{1}{2}\left\{L_j^\dagger L_j,N\right\}_+.
	\end{align}
	One can directly check $\cLd(I)=0$. Besides $e^{t\cLt}$ is CPTP by Theorem \ref{thm:CPTP} by setting $H_{system}=0$. 
	
	For (\romannumeral2): 
	we have
\begin{align}
	\cLt^{*}(N) &=  \sigma^{-\frac{1}{2}} \cLd \left( \sigma^{\frac{1}{2}}N \sigma^{\frac{1}{2}}\right)\sigma^{-\frac{1}{2}}\\
&=\sum_{j} \sigma^{-\frac{1}{2}}L_j^\dagger \sigma^{\frac{1}{2}} \cdot N \cdot  \sigma^{\frac{1}{2}} L_j  \sigma^{-\frac{1}{2}} - \sigma^{-\frac{1}{2}}\frac{1}{2}\left\{L_j^\dagger L_j,\sigma^{\frac{1}{2}}N\sigma^{\frac{1}{2}}\right\}_+ \sigma^{-\frac{1}{2}}\\
&=  \sum_{j} \sigma^{-\frac{1}{2}}L_j^\dagger \sigma^{\frac{1}{2}} \cdot N \cdot  \sigma^{\frac{1}{2}} L_j  \sigma^{-\frac{1}{2}} - \frac{1}{2}\left(  \sigma^{-\frac{1}{2}} L_j^\dagger L_j \sigma^{\frac{1}{2}}\cdot N + N\cdot \sigma^{\frac{1}{2}} L_j^\dagger L_j\sigma^{-\frac{1}{2}}\right).\label{eq:Ls1}
\end{align}
We will prove by properly defining a Hermitian $H_{system}$, the $\cLt^*$ can be written in the form of Eq.~(\ref{eq:Lstandard}) in Theorem \ref{thm:CPTP},  which will imply $\cLt^*$ is CPTP. Denote 
\begin{align}
	&O_j :=  \sigma^{-\frac{1}{2}} L_j^\dagger  \sigma^{\frac{1}{2}}, \quad K:=\frac{1}{2}\,\sigma^{-\frac{1}{2}} L_j^\dagger L_j \sigma^{\frac{1}{2}}, \label{eq:Ls2}\\
	&V :=  \frac{1}{2}\sum_j O_j^\dagger O_j, \quad 
	H_{system}:= \frac{1}{2i}\left(K-K^\dagger\right) \label{eq:Ls3}
\end{align}
Recall that $\cLt(\rho_\cLt)=0$ and $\sigma=\rho_\cLt^{-1}$, we have
\begin{align}
&\sum_j L_j \sigma^{-1} L_j^\dagger -\frac{1}{2}\left(L_j^\dagger L_j \sigma^{-1} +	 \sigma^{-1} L_j^\dagger L_j\right)=0. \\
\Rightarrow & \sigma^{\frac{1}{2}}\left(\sum_j L_j \sigma^{-1} L_j^\dagger -\frac{1}{2}\left(L_j^\dagger L_j \sigma^{-1} +	 \sigma^{-1} L_j^\dagger L_j\right)\right)  \sigma^{\frac{1}{2}}=0.\\
\Rightarrow & \sum_j  \sigma^{\frac{1}{2}} L_j  \sigma^{-\frac{1}{2}}\cdot  \sigma^{-\frac{1}{2}} L_j^\dagger  \sigma^{\frac{1}{2}}-\frac{1}{2}\left( \sigma^{\frac{1}{2}} L_j^\dagger L_j \sigma^{-\frac{1}{2}} +	 \sigma^{-\frac{1}{2}} L_j^\dagger L_j \sigma^{\frac{1}{2}}\right)  =0.
\end{align}
Thus 
\begin{align}
&V= \frac{1}{2}\left(K^\dagger +K\right)\quad \Rightarrow\quad K = V + iH_{system}.\label{eq:Ls4}
\end{align}
where the $\Rightarrow$ we use $H_{system}:= \frac{1}{2i}\left(K-K^\dagger\right)$.
One can verify
\begin{align}
\cLt^*(N) &= \sum_j O_j NO_j^\dagger - \left(KN+NK^\dagger\right)\\
&= \sum_j O_j NO_j^\dagger - \left(VN+NV\right) + i[N,H_{system}]\\
&=i[N,H_{system}] +\sum_j O_j N O_j^\dagger -\frac{1}{2}\left\{O_j^\dagger O_j,N\right\}_+	 \label{eq:Lstar_Lind}
\end{align}
where the first equality comes from Eqs.~(\ref{eq:Ls1})(\ref{eq:Ls2})(\ref{eq:Ls3}), the second equality comes from (\ref{eq:Ls4}), and the last equality comes from definition of $V$ in Eq.~(\ref{eq:Ls3}).
One can check that $H_{system}$ is Hermitian. Since $\cLs=\frac{1}{2}(\cLt+\cLt^*)$, from Eqs.~(\ref{eq:Lss})(\ref{eq:Lstar_Lind}) we known that $\cLs$ can also be written as Eq.~(\ref{eq:Lstandard}), where the Lindbladian jump operator is $\{\frac{1}{\sqrt{2}}L_j\}\cup \{\frac{1}{\sqrt{2}}O_j\}$. Thus by Theorem \ref{thm:CPTP} we conclude that $e^{t\cLs}$ is CPTP for $t\geq 0$. 

For (\romannumeral3): From Assumption \ref{ass:L} we already have $\{\frac{1}{\sqrt{2}}L_j\}$  generates the full algebra. Thus so does $\{\frac{1}{\sqrt{2}}L_j\}\cup \{\frac{1}{\sqrt{2}}O_j\}$.  By Theorem \ref{thm:wol12} we conclude that there exists a unique $\rho$ such that $\cLs(\rho)=0$. Besides, from Eq.~(\ref{eq:Ls1}) one can directly verify that 
	$\cLt^*(\rho_\cLt)=0$ thus 	$\cLs(\rho_\cLt)=0$.

For (\romannumeral4): By definition of $\cLs:\cH(n)\rightarrow\cH(n)$ it is Hermitian w.r.t. $\langle,\rangle_\sigma$. Thus $\cLs$ is diagonalizable and the spectrum is real. With contradiction assume  there exists an eigenstate $N\in \cH(n), N\neq 0$ with eigenvalue $\lambda>0$. Then $\tn{e^{t\cLs}(N)}=e^{\lambda t}\tn{N}$ goes to infinity as $t\rightarrow \infty$, which contradicts with the fact that $e^{t\cLs}$ is CPTP from (\romannumeral2), where CPTP map does not increase trace norm of Hermitian operators (Lemma \ref{lem:CPTP} in Appendix \ref{appendix:Matrix_Norm}).

For (\romannumeral5): By (\romannumeral3)(\romannumeral4) we know the spectral gap is strictly greater than 0, and the minimum eigenvector is $\rho_\cLt$ and has eigenvalue 0. Then Eq.~(\ref{eq:specGap}) is just by definition of spectral gap.  
\end{proof}

\subsection{Proof of Theorem \ref{thm:mix_Ls_gap}}
\begin{proof}[of Theorem \ref{thm:mix_Ls_gap}]
All the properties of $\cLt$ we use are summarized in Assumption \ref{ass:L}.
	Define
$$h_t	 =e^{t\cLt}(\rho) -\rho_\cLt.$$
Since $\cLt(\rho_\cLt)=0$, one can check that
\begin{align}
	\frac{d}{dt} \langle h_t,h_t\rangle_\sigma = 2\left\langle h_t,\cLs(h_t)\right\rangle_\sigma.	\label{eq:ht}
\end{align}
Further note that since $e^{t\cLt}$ is trace-preserving from Theorem \ref{thm:assumptions} (\romannumeral1), we have $tr(e^{t\cLt}(\rho))=1$. One can verify that  $
\langle h_t,\rho_\cLt\rangle_\sigma = 0.$
Then by Eq.~(\ref{eq:ht}) and Eq.~(\ref{eq:specGap}), we have 
\begin{align}
	\frac{d}{dt} \langle h_t,h_t\rangle_\sigma &= 2 \left\langle \frac{h_t}{\sn{h_t}{\sigma}},\cLs\left(\frac{h_t}{\sn{h_t}{\sigma} } \right)\right\rangle_\sigma \cdot \langle h_t,h_t\rangle_\sigma\\
	&\leq -2 \Upsilon \cdot \langle h_t,h_t\rangle_\sigma.
\end{align}
Then by Gronwall's inequality we have
\begin{align}
	&\langle h_t,h_t\rangle_\sigma \leq \langle h_0,h_0\rangle_\sigma \cdot \exp(-2\Upsilon t ).\label{eq:decay}
\end{align}
Recall that $h_0=\rho-\rho_\cLt$. Thus we have
\begin{align}
\tn{e^{t\cLt}(\rho)-\rho_\cLt}	&\leq  2^{n/2} \sn{e^{t\cLt}(\rho)-\rho_\cLt}{\sigma}\\
& \leq  2^{n/2} \sn{\rho-\rho_\cLt}{\sigma} \cdot \exp(-\Upsilon t )\\
&\leq  2^{n/2}  \cdot \sqrt{tr(\sigma^{\frac{1}{2}}\rho\sigma^{\frac{1}{2}}\rho)}\cdot \exp(-\Upsilon t ).
\end{align}
where the first inequality comes from connecting trace-norm and $\sn{\cdot}{\sigma}$ by Lemma \ref{lem:sig_tn} in Appendix \ref{appendix:Matrix_Norm}. Note that $\sigma\geq I$ since $\sigma=\rho_\cLt^{-1}$ and $\rho_\cLt$ is  a quantum state. The second inequality comes from Eq.~(\ref{eq:decay}). The last inequality comes from $\sigma=\rho_\cLt^{-1}$ and $tr(\rho-\rho_\cLt)=0$ since both $\rho$ and $\rho_\cLt$ are quantum states, and 
\begin{align}
    \sn{\rho-\rho_\cLt}{\sigma}^2 = tr(\sigma^{\frac{1}{2}} (\rho-\rho_\cLt) \sigma^{\frac{1}{2}}(\rho-\rho_\cLt)) = tr(\sigma^{\frac{1}{2}}\rho \sigma^{\frac{1}{2}} \rho)-1.
\end{align}
\end{proof}

\bibliographystyle{alpha}
\bibliography{ref.bib}

\end{document}

%% file: style.tex
\usepackage[utf8]{inputenc}
\usepackage{amsmath}
\usepackage[colorlinks,linkcolor=blue,citecolor=red]{hyperref}
\usepackage{xcolor}
\usepackage{graphicx}
\usepackage{amsfonts}
\usepackage{algpseudocode}
\usepackage{amssymb}
\usepackage{comment}
\usepackage{authblk}

\usepackage[paperwidth=199.8mm,paperheight=287mm,centering,hmargin=25mm,vmargin=2cm]{geometry}
\usepackage{theorem}
\usepackage{ulem}

\newcommand{\norm}[1]{\left|#1\right|_{\star}}
\newcommand{\flo}[1]{\lfloor #1\rfloor}
\newcommand{\cei}[1]{\lceil #1\rceil}
\newcommand{\cg}[0]{  g/5} 
\newcommand{\cgh}[0]{  g/10}

\newcommand{\bL}[0]{ 2^{-\cgh +2n+4}+ 40\beta\cdot  \kappa_H\cdot  2^{-r}}
\newcommand{\nI}[0]{T}
\newcommand{\tn}[1]{\left|#1\right|_{1}}
\newcommand{\sn}[2]{\left|#1\right|_{#2}}

\newcommand{\Mod}[0]{\!\!\!\mod}
\newcommand{\qpe}[0]{\bm{\rm{QPE}}}
\newcommand{\fqpe}[0]{\bm{\rm{FQPE}}}

\newcommand{\QPE}[0]{\bm{\rm{BQPE}}}
\newcommand{\sQPE}[0]{\rm{QPE}}
\newcommand{\FQPE}[0]{\bm{\rm{FBQPE}}}

\newcommand{\bb}[0]{\bm{b}}
\newcommand{\RM}[1]{\uppercase\expandafter{\romannumeral#1}}

\newcommand{\acc}[0]{{\sc{Accept}}}
\newcommand{\aacc}[0]{{\sc{AltAccept}}}
\newcommand{\rej}[0]{{\sc{Reject}}}

\newcommand{\cLd}[0]{\mathcal{L}^{\diamond}}

\newcommand{\cLs}[0]{\mathcal{L}^{(s)}}
\newcommand{\cLt}[0]{\mathcal{L}}

\newcommand{\cP}[0]{\mathcal{P}}
\newcommand{\cJ}[0]{\mathcal{J}}

\newcommand{\cJt}[0]{\mathcal{J}}

\newcommand{\cI}[0]{\mathcal{I}}

\newcommand{\cD}[0]{\mathcal{D}}
\newcommand{\cET}[0]{\mathcal{E}[\tau]}
\newcommand{\cEN}[0]{\mathcal{E}}

\newcommand{\cC}[0]{\mathcal{C}}

\newcommand{\cH}[0]{\mathcal{H}}

\newcommand{\cR}[0]{\mathcal{R}}

\newcommand{\cF}[0]{\mathcal{F}}

\newcommand{\cM}[0]{\mathcal{M}}

\newcommand{\bm}[1]{\boldsymbol{#1}}

\newcommand{\ket}[1]{\left\vert #1 \right\rangle}
\newcommand{\bra}[1]{\left\langle #1 \right\vert}
\newcommand{\ketbra}[2]{\vert#1\rangle\!\langle#2\vert}

\newcommand{\bC}[0]{\mathbb{C}}

\newcommand{\bR}[0]{\mathbb{R}}
\newcommand{\bN}[0]{\mathbb{N}}

\newtheorem{definition}{Definition}

\newtheorem{assumption}{Assumption}

\newtheorem{lemma}[definition]{Lemma}
\newtheorem{algorithm}[definition]{Algorithm}

\newtheorem{theorem}[definition]{Theorem}
\newtheorem{corollary}[definition]{Corollary}

\newenvironment{proof}{{\bf Proof:}}{\hfill\rule{2mm}{2mm}}

\usepackage{algorithm,algpseudocode}